\documentclass[11pt, a4paper, twoside]{amsart}
\usepackage{mathrsfs, tabularx}

\usepackage{fancybox}   

\usepackage[margin=3.5cm]{geometry}

\usepackage {amsthm, amsmath, amssymb, amscd, float, latexsym, times, epic, eepic}

\theoremstyle{plain}
\newtheorem{theorem}{Theorem}[section]
\newtheorem{lemma}[theorem]{Lemma}
\newtheorem{proposition}[theorem]{Proposition}

\theoremstyle{definition}
\newtheorem{definition}[theorem]{Definition}

\newcommand {\Set}[1] {\mathbb{#1}}
\newcommand{\setR}[0]{\Set{R}}

\newcommand{\setC}[0]{\Set{C}}

\newcommand{\mediumMatrix}[0]{P}
\newcommand {\newdecomp}[0] {algebraically decomposable}

\newcommand{\cG}[0]{{\mathscr{G}}}

\newcommand{\rr}[0]{\alpha}
\newcommand{\pos}[0]{\beta}

\newcommand{\pd}[2]{\frac{\partial #1}{\partial #2}}

\newcommand{\medTensor}[0]{{\Omega^2\ \!\!_2}}
\newcommand{\medTwTensor}[0]{{\widetilde \Omega^2\ \!\!_2}}

\newcommand{\slaz}[0]{\setminus \{0\}}

\title[{Characterisation and representation of medium with a double light cone} ]
{Characterisation and representation of non-dissipative electromagnetic medium with a double light cone}
\author[Dahl]{Matias F. Dahl}

\begin{document}
\maketitle
\begin{abstract}
  We study Maxwell's equations on a $4$-manifold $N$ with a medium that is non-dissipative and has a
  linear and pointwise response. In this setting, the medium can be represented by a suitable
  $2\choose 2$-tensor on the $4$-manifold $N$.  Moreover, in each cotangent space on $N$, the medium
  defines a \emph{Fresnel surface}.  Essentially, the Fresnel surface is a tensorial analogue of the
  dispersion equation that describes the response of the medium for signals in the geometric optics
  limit. For example, in isotropic medium the Fresnel surface is at each point a Lorentz light cone.
  In a recent paper, I.~Lindell, A.~Favaro and L.~Bergamin introduced a condition that constrains
  the polarisation for plane waves.  In this paper we show (under suitable assumptions) that a
  slight strengthening of this condition gives a pointwise characterisation of all medium tensors
  for which the Fresnel surface is the union of two distinct Lorentz null cones. This is for example
  the behaviour of uniaxial medium like calcite. Moreover, using the representation formulas from
  Lindell \emph{et al.} we obtain a closed form representation formula that pointwise parameterises
  all medium tensors for which the Fresnel surface is the union of two distinct Lorentz null cones.
  Both the characterisation and the representation formula are tensorial and do not depend on local
  coordinates.
\end{abstract}


\section{Introduction}

We will study the \emph{pre-metric} Maxwell's equations, where Maxwell's equations are written on a
$4$-manifold $N$ and the electromagnetic medium is described by a suitable antisymmetric $2\choose
2$-tensor $\kappa$ on $N$ that pointwise is determined by $36$ real parameters.
%
In each cotangent space on $N$, the electromagnetic medium determines a fourth order polynomial
surface called the \emph{Fresnel surface} that can be seen as a tensorial analogue of the dispersion
equation. The Fresnel surface describes the response of the medium to signals in the geometric
optics limit \cite{ ObuFukRub:00, Rubilar2002, Obu:2003, PunziEtAl:2009, RSS:2011}.
In this work we will assume that the medium is \emph{skewon-free}. Then there are only $21$ free
parameters and such medium models non-dissipative medium. For example, under suitable assumptions
the skewon-free assumption will imply that Poynting's theorem holds \cite{Obu:2003, Dahl:2009}.
%
%
%
On an orientable manifold one can show that invertible skewon-free $2\choose 2$-tensors are in
one-to-one correspondence with \emph{area metric}. By an area metric, we here mean a $0\choose
4$-tensor on $N$ that defines a symmetric non-degenerate inner product for bivectors.  Area metrics
appear when studying the propagation of a photon in a vacuum with a first order correction from
quantum electrodynamics \cite{DruHath:1980, Schuller:2010}.  The Einstein field equations have also
been generalised into equations where the unknown field is an area metric \cite{PSW_JHEP:2007}.  For
further examples, see \cite{PunziEtAl:2009, Schuller:2010}.

We know that in isotropic medium like vacuum, the Fresnel surface is a Lorentz null cone at each
point in $N$. That is, Lorentz geometry describes the propagation of light in isotropic medium.
Conversely, it was conjectured in $1999$ by Y.~Obukhov and F.~Hehl \cite{ObukhovHehl:1999,
  ObuFukRub:00} that isotropic medium is the only (non-dissipative and axion-free) medium where the
Fresnel surface is a Lorentz null cone.
This was partially proven already in  \cite{ObuFukRub:00}. However, the full conjecture 
was only established in \cite{FavaroBergamin:2011} by A.~Favaro and L.~Bergamin.
For an alternative proof, see \cite{Dahl:2011:Closure} and for 
further discussions and related results, see \cite{ObuRub:2002, Obu:2003, LamHeh:2004, Itin:2005} 
and Section \ref{sec:nonbire} below.

Since the Fresnel surface is a $4$th order polynomial surface, the Fresnel
surface can also decompose into the union of two distinct Lorentz null cones.  For example, this is
the case in \emph{uniaxial medium} like calcite (CaCO$_3$) \cite[Section 15.3]{BornWolf:1999}.  In
such medium, the propagation properties of the medium does not only depend on direction, but also on
the polarisation of the wave. In uniaxial medium, there are two eigenpolarisations and one null cone
for each polarisation. In consequence, there is one Fermat's principle for each polarisation
\cite{PunziEtAl:2009}.  This is the the source for the physical phenomenon of double refraction.

We know that \emph{uniaxial medium} is an example of medium with two distinct null cones.  A
natural next task is to understand the structure of all medium tensors with this property.  This is
the main result in \cite{Dahl:2011:DoubleCone}, which gives the complete local description of all
non-dissipative medium tensors for which the Fresnel surface is a double light cone (up to suitable
assumptions).  The importance of this result is that it shows that are three and only three medium
classes with this behaviour. Moreover, the theorem gives explicit coordinate expressions for each
medium class.  The first medium class is a slight generalisation of uniaxial medium. The second
class seems to be a new class of mediums. The last class seems to be unphysical; heuristic arguments
and preliminary numerical tests suggest that Maxwell's equations are not hyperbolic in that class
\cite{Dahl:2011:DoubleCone}.  In the below, this result is summarised in Theorem
\ref{thm:factorMedium}.

The main contribution of this paper is Theorem \ref{thm:MainDecomp}.  Under suitable assumptions,
this theorem gives a tensorial characterisation (condition \ref{thm:MainDecomp:II} in Theorem
\ref{thm:MainDecomp}) of all non-dissipative medium tensors for which the Fresnel surface is two
distinct light cones. In a suitable limit, the condition also reduces to the \emph{closure
  condition} $\kappa^2 = -\lambda\operatorname{Id}$ for a $\lambda>0$ that characterises medium with
a single light cone \cite{Obu:2003}. Moreover, in Theorem \ref{thm:MainDecomp} we give a tensorial
representation formula (equation \eqref{eq:RepresentationOfDoubleCone}) that parameterises all
non-dissipative medium tensors with two distinct light cones. Both the characterisation and
representation formula are pointwise results.

The background and motivation for Theorem \ref{thm:MainDecomp} comes from a recent paper by
I.~Lindell, A.~Favaro and L.~Bergamin \cite{LindellBergaminFavaro:Decomp}.  In Section
\ref{eq:twoChar} we will briefly summarise some of the results from
\cite{LindellBergaminFavaro:Decomp}.  In this paper, the authors introduces a second order polynomial
condition on the medium tensor (equation \eqref{eq:eq44} in the below).
Equation \eqref{eq:eq44} is derived from a constraint on polarisation of plane waves, and in
\cite{LindellBergaminFavaro:Decomp} it is shown that whenever condition \eqref{eq:eq44} is satisfied
(plus some additional assumptions), the Fresnel surface always factorises into two second order
surfaces.
In Section \ref{sec:44Factorisability} we will further motivate that equation \eqref{eq:eq44} is in fact a
general factorisability condition for the Fresnel surface.  At first this might seem unexpected
since equation \eqref{eq:eq44} was initially derived from a constraint on polarisation, yet it is
able to constrain the behaviour of signal speed. However, the explanation is that for
electromagnetic waves, polarisation and signal speed are not independent properties but tied
together.  In Theorem \ref{thm:MainDecomp}, condition \ref{thm:MainDecomp:II} is a slight
strengthening of equation \eqref{eq:eq44}.  Also, representation formula 
\eqref{eq:RepresentationOfDoubleCone} in Theorem
\ref{thm:MainDecomp} is adapted from \cite{LindellBergaminFavaro:Decomp} and constitute a subclass
of generalised $Q$-medium introduced by I.~Lindell and H.~Wall\'en in \cite{LindellWallen:2002}.
A further technical discussion on Theorem \ref{thm:MainDecomp} is given in the end of Section
\ref{sec:DoubleCones}.

Some of the computations in the paper rely on computer algebra. For further
information about the Mathematica notebooks for these computations, please see the author's homepage.

\section{Preliminaries}
\label{sec:prelim}
By a \emph{manifold} $N$ we mean a second countable topological Hausdorff
space that is locally homeomorphic to $\setR^n$ with $C^\infty$-smooth
transition maps. All objects are assumed to be smooth where defined.  
Let $TN$ and $T^\ast N$ be the tangent and cotangent bundles, respectively. 
For $k\ge 1$, let $\Omega^k(N)$ be antisymmetric tensor fields with $k$ lower indices (that is,
$k$-forms).  Similarly, let $\Omega_k(N)$ be antisymmetric tensor fields with $k$ upper
indices. Moreover, let $\medTensor(N) = \Omega^2(N)\otimes \Omega_2(N)$.  Let also $C^\infty(N)$ be
the set of scalar functions (that is, $0\choose 0$-tensors).  The Einstein summing convention is used
throughout. When writing tensors in local coordinates we assume that the components satisfy the same
symmetries as the tensor.
 
\subsection{Twisted tensors}
If $N$ is not orientable we will also need \emph{twisted tensors} \cite[Section A.2.6]{Obu:2003}.
We will denoted these by a tilde over the tensor space.  For example, by $\widetilde \Omega^2(N)$ we
denote the space of twisted $2$-forms.  If $G\in \widetilde \Omega^2(N)$ then in each coordinate
chart $(U, x^i)$, $G$ is determined by a usual $2$-form $G\vert_U\in \Omega^2(U)$ 
and on overlapping charts $(U,x^i)$ and $(\widetilde U,
\widetilde x^i)$, forms $G\vert_U$ and $G\vert_{\widetilde U}$ satisfy the transformation
rule
\begin{eqnarray}
\label{eq:FrankelFormTransform}
   G\vert_{\widetilde U} &=&\operatorname{sgn} \det \left( \pd{  x^a}{\widetilde x^b}\right) G\vert_U,
\end{eqnarray}   
where 
$\operatorname{sgn}\colon \setR\to \setR$ is the \emph{sign function}, $\operatorname{sgn} x= x/\vert x\vert$ for
$x\neq 0$ and $\operatorname{sgn}x=0$ for $x=0$.  
If locally
\begin{eqnarray}
\label{eq:GUU}
G\vert_U &=& \frac 1 2 G_{ij} dx^i \wedge dx^j, \quad
G\vert_{\widetilde U} = \frac 1 2 \widetilde G_{ij} d\widetilde x^i \wedge d\widetilde x^j,
\end{eqnarray} 
then equation \eqref{eq:FrankelFormTransform} implies that components $G_{ij}$ and $\widetilde
G_{ij}$ transform as
\begin{eqnarray}
\label{eq:FTwistRule}
\widetilde G_{ij} &=& \operatorname{sgn} \det \left( \pd{ x^a}{\widetilde x^b}\right) 
G_{rs} \pd{x^r}{\widetilde x^i}\pd{x^s}{\widetilde x^j}. 
\end{eqnarray}
When the chart is clear from context, we will simply write $G=\frac 1 2 G_{ij} dx^i \wedge dx^j$.
Similarly, if $\kappa\in \medTwTensor(N)$ then in each chart $\kappa$ is represented by a 
$\kappa\vert_U \in \medTensor(U)$ and locally 
\begin{eqnarray}
\label{eq:kappaLocal}
\kappa &=& \frac 1 8 \kappa^{ij}_{rs} dx^r\wedge dx^s\otimes \pd{}{x^i}\wedge \pd{}{x^j}
\end{eqnarray}
for suitable components $\kappa^{ij}_{rs}$.
Moreover, 
if $\kappa^{ij}_{rs}$ and $\widetilde \kappa^{ij}_{rs}$ are components for $\kappa$ in overlapping
charts $(U, x^i)$ and $(\widetilde U, \widetilde x^i)$ then we obtain the transformation rule
\begin{eqnarray}
\label{eq:kappaTwistRule} 
\widetilde \kappa^{ij}_{rs} &=& \operatorname{sgn} \det \left( \pd{ x^a}{\widetilde x^b}\right) 
\kappa^{pq}_{uv} \pd{x^u}{\widetilde x^r}\pd{x^v}{\widetilde x^s}
\pd{\widetilde x^i}{ x^p}\pd{\widetilde x^j}{x^q}.
\end{eqnarray}
 
Compositions involving twisted tensors
are computed in the
natural way by composing local tensors. 
%
%
For example, if $\kappa, \eta\in \medTwTensor(N)$ their composition defines an element $\kappa\circ \eta
\in \medTensor(N)$ and if $\kappa, \eta$ and  $\kappa\circ \eta$ are written as in 
equation \eqref{eq:kappaLocal} then
\begin{eqnarray}
(\kappa\circ\eta)^{ij}_{rs} &=& \label{eq:kappaEtaProd}
\frac 1 2 \kappa^{ab}_{rs} \eta^{ij}_{ab}.
\end{eqnarray}
%
%
%

If $M$ is orientable, then twisted tensors coincide with their normal (or untwisted) counterparts.
For example, if $M$ is orientable, equation \eqref{eq:kappaTwistRule} implies that
$\medTwTensor(N)= \medTensor(N)$.
%
There are also other way to define twisted forms.  Equation
\eqref{eq:FrankelFormTransform} coincides with definition of a pseudo-form in
\cite{Frankel:2004}.  For a global definition of twisted forms using the orientation bundle,
see \cite[Supplement 7.2A]{AbrahamMarsdenRatiu:2001}.

\subsection{Tensor densities}
In addition to tensors and twisted tensors, we will need tensor densities and twisted
tensor densities. 
%
A $p\choose q$-\emph{tensor density of weight} $w\in \mathbb{Z}$ on a manifold $N$ is determined by
components $T^{a_1 \ldots a_p}_{b_1 \cdots b_q}$ in each chart $(U,x^i)$, and on overlapping charts
$( U, x^i)$ and $(\widetilde U, \widetilde x^i)$ we have the transformation rule
\cite{Spivak:I:1999},
\begin{eqnarray*}
\widetilde T^{a_1 \ldots  a_p}_{b_1 \cdots b_q} &=& 
\left( \det \left( \pd{x^i}{\widetilde x^j}\right) \right)^w\,\,
T^{r_1 \ldots  r_p}_{s_1 \cdots s_q} 
\pd{x^{s_1}}{\widetilde x^{b_1}} \cdots\pd{x^{s_q}}{\widetilde x^{b_q}}
\pd{\widetilde x^{a_1}}{x^{r_1}} \cdots\pd{\widetilde x^{a_p}}{x^{r_p}}.
\end{eqnarray*}

A \emph{twisted} $p\choose q$-\emph{tensor density of weight} $w\in \mathbb{Z}$ on $N$ is defined in the
same way, but with an additional $\operatorname{sgn} \det \left( \pd{\widetilde x^i}{x^j}\right)$
factor in the transformation rule as in equations \eqref{eq:FTwistRule} and
\eqref{eq:kappaTwistRule}.
 
The \emph{Levi-Civita permutation symbols} are denoted by $\varepsilon_{ijkl}$ and 
$\varepsilon^{ijkl}$.  
Even if these coincide as combinatorial functions so that $\varepsilon_{ijkl} =
\varepsilon^{ijkl}$, they are also different as they globally define different objects on a
manifold. Namely, if $\varepsilon_{ijkl}, \varepsilon^{ijkl}$ and $\widetilde \varepsilon_{ijkl},
\widetilde \varepsilon^{ijkl}$ are defined on overlapping coordinate charts $(U,x^i)$ and
$(\widetilde U,\widetilde x^i)$, respectively, then 
\begin{eqnarray}
\widetilde \varepsilon_{abcd} &=& \label{eq:epsTransA}
\operatorname{det}\left( \pd{\widetilde x^i}{ x^j}\right)  \varepsilon_{pqrs}
\pd{ x^p}{\widetilde x^a}
\pd{x^q}{\widetilde  x^b}
\pd{x^r}{\widetilde  x^c}
\pd{ x^s}{\widetilde  x^d}, \\
\widetilde  \varepsilon^{abcd} &=& \label{eq:epsTransB}
\operatorname{det}\left( \pd{ x^i}{\widetilde  x^j}\right)  \varepsilon^{pqrs}
\pd{\widetilde  x^a}{ x^p}
\pd{ \widetilde x^b}{ x^q}
\pd{ \widetilde x^c}{  x^r}
\pd{\widetilde x^d} {  x^s}.
\end{eqnarray}
That is, $\varepsilon_{ijkl}$ defines a $0\choose 4$-tensor density of weight $-1$ on $N$
and  $\varepsilon^{ijkl}$ defines a $4\choose 0$-tensor density of weight $1$.
For future reference, let us note that
\begin{eqnarray}
\label{eq:epsIdentities}
\varepsilon^{rsab} \varepsilon_{rsij} = 4 \delta^a_{[i}\delta^b_{j]}, \quad
\varepsilon^{rabc} \varepsilon_{rijk} = 3! \delta^a_{[i}\delta^b_{j}\delta^c_{k]},
\end{eqnarray}
where $\delta^i_j$ is the \emph{Kronecker delta symbol} and
brackets $[i_1 \ldots i_p]$ indicate that indices $i_1, \ldots, i_p$ are antisymmetrised
with scaling $1/p!$. 
\subsection{Maxwell's equations on a $4$-manifold}
\label{sec:MaxOn4}
On a $4$-manifold $N$, the \emph{premetric Maxwell's equations} read 
\begin{eqnarray}
\label{max4A}
dF &=& 0, \\
\label{max4B}
dG &=& J,\\
\label{FGchi}
 G &=& \kappa(F).
\end{eqnarray}
where $d$ is the exterior derivative, $F\in \Omega^2 (N)$, $G\in \widetilde \Omega^2(N)$, $J\in
\widetilde \Omega^3(N)$ and $\kappa\in \medTwTensor(N)$.  Here, $F, G$, are called the \emph{
  electromagnetic field variables}, $J$ describes the electromagnetic sources, 
tensor $\kappa$ models the electromagnetic medium and equation
\eqref{FGchi} is known as the \emph{constitutive equation}.
In local coordinates, equations \eqref{max4A}--\eqref{FGchi} reduce to the usual Maxwell's
equations.  For a systematic treatment, see \cite{Rubilar2002, Obu:2003}.

If locally $F = \frac 1 2 F_{ij} dx^i \wedge dx^j, G = \frac 1 2 G_{ij} dx^i \wedge dx^j$ and
$\kappa$ is written as in equation \eqref{eq:kappaLocal} then constitutive equation \eqref{FGchi} is
equivalent with
\begin{eqnarray}
\label{FGeq_loc}
 G_{ij} &=& \frac 1 2 \kappa_{ij}^{ab} F_{ab}.
\end{eqnarray}
Thus  equation \eqref{FGchi} models electromagnetic medium with a linear and pointwise
response.  

Suppose $\kappa\in \medTwTensor(N)$ and suppose $(U,x^i)$ is a chart.  Then the local
representation of $\kappa$ in equation \eqref{eq:kappaLocal} defines a pointwise linear map $\Omega^2(U)\to
\Omega^2(U)$.  In $U$ we can therefore represent $\kappa$ by a smoothly varying $6\times 6$ matrix.
To do this, let $O$ be the ordered set of index pairs $\{ 01, 02, 03$, $23, 31, 12\}$, and if $J\in
O$, let $dx^J = dx^{J_1}\wedge dx^{J_2}$, where  $J_1$ and $J_2$ are the individual indices for $J$. 
Say, if $J=31$ then $dx^J = dx^3\wedge dx^1$. Then a basis for $\Omega^2(U)$ is given by $\{ dx^{J}:
J\in O\}$, that is,
\begin{eqnarray}
\label{eq:2basis}
\{ dx^0\wedge dx^1, dx^0\wedge dx^2, dx^0\wedge dx^3, dx^2\wedge dx^3, dx^3\wedge dx^1, dx^1\wedge dx^2\}.
\end{eqnarray}
This choice of basis follows \cite[Section A.1.10]{Obu:2003}.
By equation \eqref{eq:kappaLocal} it follows that
\begin{eqnarray}
\label{eq:kappaMatDef}
\kappa(dx^{J}) &=& \sum_{I\in O} \kappa^J_I dx^I, \quad J \in O,
\end{eqnarray} 
where $\kappa^J_I = \kappa^{J_1 J_2}_{I_1 I_2}$. 
Let $b$ be
the natural bijection $b\colon O\to \{1,\ldots, 6\}$.  Then we
identify coefficients $\{\kappa^J_I : I,J\in O\}$ for $\kappa$ with the smoothly varying $6\times 6$
matrix $\mediumMatrix=(\kappa^J_I)_{IJ}$ defined as $\kappa^J_I = \mediumMatrix_{b(I) b(J)}$
for $I,J\in O$.

Suppose $\mediumMatrix=(\kappa_I^J)_{IJ}$ and $\widetilde \mediumMatrix=(\widetilde
\kappa_I^J)_{IJ}$ are smoothly varying $6\times 6$ matrices that represent tensor $\kappa$ in
overlapping charts $(U, x^i)$ and $(\widetilde U, \widetilde x^i)$. Then equation
\eqref{eq:kappaTwistRule} is equivalent with
\begin{eqnarray*}
\label{eq:kappaTransRule}
\widetilde \kappa^J_I &=& \operatorname{sgn} \det \left( \pd{x^i}{\widetilde x^j}\right)
\sum_{K,L\in O} 
\pd{x^K}{\widetilde x^I}
\kappa^L_K   
\pd{\widetilde x^J}{x^L}, \quad I,J\in O,
\end{eqnarray*} 
where
\begin{eqnarray}
\label{eq:antiSymJacobian}
  \pd{x^J}{\widetilde x^I} &=&
  \pd{x^{J_1}}{\widetilde x^{I_1}}\pd{ x^{J_2}}{\widetilde  x^{I_2}}-\pd{ x^{J_2}}{\widetilde x^{I_1}}\pd{ x^{J_1}}{\widetilde x^{I_2}}, \quad I,J\in O,
\end{eqnarray}
and $\pd{\widetilde x^J}{ x^I}$ is defined similarly by exchanging $x$ and $\widetilde x$. 
%
For matrices $T=(\pd{x^J}{\widetilde x^I})_{IJ}$ and $S=(\pd{\widetilde x^J}{ x^I})_{IJ}$, we have
$T=S^{-1}$, whence equation \eqref{eq:kappaTwistRule} is further equivalent with the matrix equation
\begin{eqnarray}
\label{eq:kappaTransRuleII}
\widetilde \mediumMatrix &=& \operatorname{sgn} \det \left( \pd{x^i}{\widetilde x^j}\right)
\,\, T \mediumMatrix  T^{-1}.
\end{eqnarray}

In a chart $(U, x^i)$, we define $\operatorname{trace} \kappa\colon U\to \setR$ and $\det \kappa\colon
U\to \setR$ as the trace and determinant of the pointwise linear map $\Omega^2(U)\to
\Omega^2(U)$. When $\mediumMatrix$ is as above it follows that $\operatorname{trace} \kappa =
\operatorname{trace} \mediumMatrix$ and $\operatorname{det} \kappa = \operatorname{det}
\mediumMatrix$.  When these definitions are extended into each chart on $N$ equation
\eqref{eq:kappaTransRuleII} shows that $\operatorname{trace} \kappa \in
\widetilde C^\infty(N)$ and $\operatorname{det} \kappa \in C^\infty(N)$. Moreover, if $\kappa$ is
written as in equation \eqref{eq:kappaLocal}, then
\begin{eqnarray*}
\operatorname{trace} \kappa &=& \frac 1 2 \kappa_{ij}^{ij}.
\end{eqnarray*}
At a point $p\in N$ we say that $\kappa$ is \emph{invertible} if $(\det \kappa)\vert_p\neq 0$.
If $\operatorname{Id}$ is the identity tensor $\operatorname{Id}\in \medTensor(N)$, then writing
$\operatorname{Id}$ as in equation \eqref{eq:kappaLocal} gives $\operatorname{Id}^{ij}_{rs}=
\delta^i_r\delta^j_s-\delta^i_s\delta^j_r$.  For $f\in \widetilde C^\infty(N)$ it follows that
$\operatorname{trace} f\operatorname{Id} = 6f$.

\subsection{Decomposition of electromagnetic medium}
\label{media:decomp}

\newcommand{\kappaI}[0]{^{(1)}\!\kappa}
\newcommand{\kappaII}[0]{^{(2)}\!\kappa}
\newcommand{\kappaIII}[0]{^{(3)}\!\kappa}
\newcommand{\notkappa}[0]{\!\not\!\kappa}

At each point of a $4$-manifold $N$, an element of $\medTwTensor(N)$ depends on $36$
parameters. Pointwise, such $2\choose 2$-tensors canonically decompose into three linear
subspaces. The motivation for this decomposition is that different components in the decomposition
enter in different parts of electromagnetics.  See \cite[Section D.1.3]{Obu:2003}.  
%


\begin{proposition}
\label{theorem:Decomp}
Let $N$ be a $4$-manifold, and let
\begin{eqnarray*}
Z &=& \{ \kappa \in \medTwTensor(N) : u\wedge \kappa(v) = \kappa(u)\wedge v \,\,\mbox{for all}\,\, u,v\in \Omega^2(N),\\
& & \quad\quad\quad\quad\quad\quad \operatorname{trace} \kappa = 0\},\\
W &=& \{ \kappa \in \medTwTensor(N) : u\wedge \kappa(v) = -\kappa(u)\wedge v \,\,\mbox{for all}\,\, u,v\in \Omega^2(N)\},\\
U &=& \{ f \operatorname{Id}\in \medTwTensor(N) : f\in \widetilde C^\infty(N) \}.
\end{eqnarray*}
Then
\begin{eqnarray}
\label{AdecompSet}
\medTwTensor(N) &=& Z\,\,\oplus\,\, W \,\,\oplus\,\, U,
\end{eqnarray}
and pointwise, $\dim Z = 20$,  $\dim W = 15$ and  $\dim U = 1$.
\end{proposition}

If we write a $\kappa\in \medTwTensor(N)$ as
$\kappa = \kappaI \,\,+ \,\,\kappaII\,\,+\,\,\kappaIII$
with $\kappaI\in Z$, $\kappaII\in W$, $\kappaIII\in U$, then we say that $\kappaI$ is the
\emph{principal part}, $\kappaII$ is the \emph{skewon part}, $\kappaIII$ is the \emph{axion part} of
$\kappa$ \cite{Obu:2003}.  For a proof of Proposition \ref{theorem:Decomp} as stated above, see
\cite{Dahl:2011:Closure}, and for further discussions, see \cite{Rubilar2002, Obu:2003,
  Favaro:2011}.


In $\medTwTensor(N)$ there is a canonical isomorphism $\medTwTensor(N)\to \medTwTensor(N)$ known as
the \emph{Poincar\'e isomorphism} \cite{Greub:1978, Favaro:2011}.  Let us first give a local
definition. If $\kappa\in \medTwTensor(N)$ on a $4$-manifold $N$, we define $\overline{\kappa}$ as
the element $\overline{\kappa} \in \medTwTensor(N)$ defined as
\begin{eqnarray}
\label{eq:overlineLocal}
\overline{\kappa}^{ij}_{rs} &=& \frac 1 4 \varepsilon_{rsab} \kappa^{ab}_{cd} \varepsilon^{cdij}
\end{eqnarray}
when $\kappa$ and $\overline{\kappa}$ are written as in equation \eqref{eq:kappaLocal}.
Equations \eqref{eq:epsTransA}--\eqref{eq:epsTransB} imply that this assignment defines an element
$\overline{\kappa}\in \medTwTensor(N)$. For $\kappa\in \medTensor(N)$ we define $\overline{\kappa}$
in the same way and we also have a canonical isomorphism $\medTensor(N) \to \medTensor(N)$.
 
The next proposition collects results for $\overline{\kappa}$. In particular, part
\ref{prop:DiamondLemma:globalDef} states that $\overline{\kappa}$ can be interpreted as a formal
adjoint of $\kappa$ with respect to the wedge product for $2$-forms.  In consequence, the Poincar\'e
isomorphism is closely related to the decomposition in Proposition \ref{theorem:Decomp}.
For example, $\kappa\in \widetilde \Omega^2_2(N)$ has only a principal part if and only if $\kappa =
\overline{\kappa}$ and $\operatorname{trace} \kappa = 0$.  For a further discussion, see
\cite{Favaro:2011}.

\begin{proposition} Suppose $N$ is a $4$-manifold and $\kappa\in \medTwTensor(N)$.
\label{prop:DiamondLemma}
\begin{enumerate}
\item \label{prop:DiamondLemma:globalDef}
$\overline{\kappa}$ is the unique $\overline{\kappa}\in \widetilde \Omega^2_2(N)$ such that
\begin{eqnarray}
\label{eq:diamondTranspose}
\kappa(u)\wedge v &=& u\wedge \overline{\kappa} (v) \quad \mbox{for all}\,\, u,v\in \Omega^2(N).
\end{eqnarray}
\item \label{prop:DiamondLemma:Id} $\overline{f\operatorname{Id}} = f\operatorname{Id}$ for all 
$f\in \widetilde C^\infty(N)$.
\item \label{prop:DiamondLemma:alg} 
$\overline{\overline{\kappa}}=\kappa$ and if $\eta\in \medTwTensor(N)$, 
then $\overline{\kappa\circ \eta}=\overline{\eta}\circ \overline{\kappa}$.
\item \label{prop:DiamondLemma:trace}
$\operatorname{trace}\overline \kappa = \operatorname{trace}\kappa$.
\item \label{prop:DiamondLemma:solve}
If $u\wedge \kappa(u)=0$ holds for all $u\in \Omega^2(N)$ then $\kappa+ \overline{\kappa}=0$.
\end{enumerate}
\end{proposition}

\begin{proof}
Part \ref{prop:DiamondLemma:globalDef} 
follows by writing out both sides in equation 
\eqref{eq:diamondTranspose}  in coordinates.
Parts \ref{prop:DiamondLemma:Id}  and \ref{prop:DiamondLemma:alg} 
 follow by part \ref{prop:DiamondLemma:globalDef}.
Part \ref{prop:DiamondLemma:trace} is a direct computation.
For part \ref{prop:DiamondLemma:solve} we have 
\begin{eqnarray*}
u\wedge (\kappa+\overline\kappa)(v) &=& \frac 1 2 \left( (u+v)\wedge\kappa(u+v)
-(u-v)\wedge\kappa(u-v)\right)
\end{eqnarray*}
for all $u,v\in \Omega^2(N)$, and the claim follows since the right hand side vanishes.
%
\end{proof}

If $\rho$ is a twisted scalar tensor density of weight $1$ on a $4$-manifold $N$ and $A,B\in
\Omega_2(N)$ then we define $\rho\, \overline{A} \otimes B$ as the twisted tensor in
$\medTwTensor(N)$ defined as follows. If locally $A=\frac 1 2 A^{ij}\pd{}{x^i}\wedge \pd{}{x^j}$ and
$B=\frac 1 2 B^{ij}\pd{}{x^i}\wedge \pd{}{x^j}$ then
\begin{eqnarray}
\label{eq:rhoABBA}
(\rho\, \overline{A}\otimes B)^{ij}_{rs} &=& \rho \varepsilon_{rsab} A^{ab} B^{ij}
\end{eqnarray} 
when $\rho\, A\otimes \overline{B}$ is written as in equation \eqref{eq:kappaLocal}.  That $\rho\,
A\otimes \overline{B}$ transforms as an element in $\medTwTensor(N)$ follows by equation
\eqref{eq:epsTransA}. 
Similarly when $\rho$ is an untwisted scalar density we define $\rho\,\overline A\otimes B\in \medTensor(N)$
by equation \eqref{eq:rhoABBA}. For both twisted and untwisted $\rho$ we have identities
%
\begin{eqnarray}
\overline{ \rho\, \overline{A}\otimes B} &=& \label{eq:barId1}
\rho\, \overline B\otimes A, \\
(  \rho\, \overline A\otimes B)\circ \kappa &=& \label{eq:barId2}
\rho\, \overline A\otimes (B\kappa), \\
\overline \kappa\circ (\rho\,\overline A\otimes B) &=& \label{eq:barId3}\rho\, \overline{(A\kappa)} \otimes B, \\
(\rho\,\overline{A}\otimes B) \circ  (\rho\,\overline B\otimes A) &=&
\label{eq:barId4} \operatorname{trace} (\rho\,\overline B\otimes B) \,\,  (\rho\,\overline A\otimes A).
\end{eqnarray}
%

In Section \ref{sec:decompo} and in the proof of Theorem \ref{thm:MainDecomp} 
we will need the following lemma.
\begin{lemma} 
\label{lemma:ABBA=0}
Suppose $N$ is a $4$-manifold and $\kappa\in \medTwTensor(N)$ is defined as 
\begin{eqnarray}
\label{eq:kappa=rhoABBA}
  \kappa &=& \rho\, \left( \overline{A}\otimes B + \overline{B}\otimes A\right) + f \operatorname{Id},
\end{eqnarray}
where $\rho$ is a scalar tensor density of weight $1$, $A,B\in
\Omega_2(N)$ and $f\in \widetilde C^\infty(N)$. Then $\kappa\vert_p=0$ at a point $p\in N$ implies that 
$f\vert_p=0$ and $\rho\vert_p=0$ or
$A\vert_p=0$ or $B\vert_p=0$.
\end{lemma}

If $\kappa$ is written as in equation \eqref{eq:kappaLocal} and $A,B$ are written as
above, then equation \eqref{eq:kappa=rhoABBA} states that
\begin{eqnarray*}
  \kappa^{ij}_{rs} &=& \rho \varepsilon_{rsab} \left( A^{ab} B^{ij} + A^{ij} B^{ab} \right) + f \operatorname{Id}^{ij}_{rs}.
\end{eqnarray*}

\begin{proof} 
By restricting the analysis to $p$ and introducing notation $A^I=A^{I_1I_2}$ and $B^I=B^{I_1I_2}$,
we obtain
\begin{eqnarray}
\label{eq:rhoABBALL}
2 \rho(A^{I} B^{J}+ A^{J} B^{I}) + f \varepsilon^{IJ}&=& 0\quad \mbox{for all}\,\, I,J\in O.
\end{eqnarray}
Setting $I=J$ and summing implies that
$\sum_{I\in O} \rho A^{I} B^{I} = 0$.
Multiplying each equation in \eqref{eq:rhoABBALL} by $A^{I} B^{J}$ and $\varepsilon^{IJ}$ and summing
$I,J$ yields two scalar equations. Eliminating $f$ from these equations gives
\begin{eqnarray*}
\rho\left( \left(\sum_{I\in O} (A^{I})^2\right) \left(\sum_{I\in O} (B^{I})^2\right) + \frac 1 3 \left( \sum_{I,J\in O} \varepsilon^{IJ} A^I B^J\right)^2\right) &=&0, 
\end{eqnarray*}
and the claim follows. %
\end{proof}

\subsection{The Fresnel surface}
\label{sec:FresnelSurface}
Let 
$\kappa\in \medTwTensor(N)$ on a $4$-manifold $N$.  If $\kappa$ is locally given by equation
\eqref{eq:kappaLocal} in coordinates $\{x^i\}$, let
\begin{eqnarray}
\label{eq:defG0ee}
\cG^{ijkl}_0 &=& \frac 1 {48} 
\kappa^{a_1 a_2}_{b_1 b_2} 
\kappa^{a_3 i}_{b_3 b_4} 
\kappa^{a_4 j}_{b_5 b_6} 
\varepsilon^{b_1 b_2 b_5 k} 
\varepsilon^{b_3 b_4 b_6 l} 
\varepsilon_{a_1 a_2 a_3 a_4}.
\end{eqnarray}
If $\{\widetilde x^i\}$ are  overlapping coordinates, then 
equations \eqref{eq:kappaTwistRule}, \eqref{eq:epsTransA} and \eqref{eq:epsTransB} 
imply that components $\cG_0^{ijkl}$ satisfy the transformation rule
\begin{eqnarray}
\label{eq:TRtrans}
\widetilde \cG_0^{ijkl}
&=& \left\vert \det \left(\pd{x^r}{\widetilde x^s}\right)\right\vert \, \cG_0^{abcd} \pd{\widetilde x^i}{x^a} \pd{\widetilde x^j}{x^b}\pd{\widetilde x^k}{x^c}\pd{\widetilde x^l}{x^d}.
\end{eqnarray}
Thus components $\cG^{ijkl}_0$ define a twisted $4\choose 0$-tensor density $\cG_0$ on $N$ of weight
$1$. The \emph{Tamm-Rubilar tensor density} \cite{Obu:2003, Rubilar2002} is the symmetric part of
$\cG_0$ and we denote this twisted tensor density by $\cG$.  In coordinates, $\cG^{ijkl} =
\cG^{(ijkl)}_0$, where parenthesis indicate that indices $ijkl$ are symmetrised with scaling $1/4!$.
If locally $\xi=\xi_i dx^i$ it follows that $\cG^{ijkl} \xi_i \xi_j \xi_k \xi_l =\cG^{ijkl}_0 \xi_i
\xi_j \xi_k \xi_l$, and we call $\cG^{ijkl} \xi_i \xi_j \xi_k \xi_l$ the \emph{Fresnel polynomial}.
The \emph{Fresnel surface} at a point $p\in N$ is defined as
\begin{eqnarray}
\label{eq:Fr}
F_p(\kappa) &=& \{\xi\in T^\ast_p(N): \cG^{ijkl} \xi_i \xi_j\xi_k\xi_l  = 0\}.
\end{eqnarray}
By equation \eqref{eq:TRtrans}, the definition of $F_p(\kappa)$ does not depend on local
coordinates. Let $F(\kappa)=\coprod_{p\in N} F_p(\kappa)$ be the disjoint union of all Fresnel
surfaces.


The Fresnel surface $F(\kappa)$ is a fundamental object when studying wave propagation in Maxwell's
equations.  Essentially, equation $\cG^{ijkl} \xi_i \xi_j\xi_k\xi_l = 0$ in equation \eqref{eq:Fr}
is a tensorial analogue to the dispersion equation that describes wave propagation in the geometric
optics limit.  Thus $F(\kappa)$ constrains possible wave speed(s) as a function of
direction.
In general the Fresnel surface $F_p(\kappa)$ is a fourth order polynomial surface in $T^\ast_p(N)$,
so it can have multiple sheets and singular points \cite{ObukhovHehl:2004}.

There are various ways to derive the Fresnel surface; by studying a propagating weak singularity
\cite{ObuFukRub:00, Rubilar2002, Obu:2003}, using a geometric optics \cite{Itin:2009,
  Dahl:2011:Closure}, or as the characteristic polynomial of the full Maxwell's equations
\cite{Schuller:2010}. 
The tensorial description of the Fresnel surface is due to Y.~Obukhov, T.~Fukui and G.~Rubilar
\cite{ObuFukRub:00}.  

\section{Results for skewon-free medium}
\label{sec:normalFormSWW}
In this section we collect a number of results for twisted skewon-free tensors
that we will need in the proof of Theorem \ref{thm:MainDecomp}.
%
%
%
%

\subsection{The normal form theorem by Schuller et al.}
\label{sec:normalSchuller}
The normal form theorem for skewon-free medium by F.~Schuller, C.~Witte and M.~Wohlfarth
\cite{Schuller:2010} shows that there exists $23$ simple matrices such that any skewon-free medium
can pointwise be transformed into one of these normal forms by a coordinate transformation plus,
possibly, a conjugation by a Hodge operator.
Next we formulate a slightly simplified version of this result that is sufficiently general for the
proof of Theorem \ref{thm:MainDecomp}.
Let us note that the original theorem in \cite{Schuller:2010} is formulated for \emph{area metrics}.
However, under mild assumptions these are essentially in one-to-one correspondence with skewon-free
tensors in $\medTensor(N)$. The below presentation is based on the reformulation in
\cite{Dahl:2011:Restatement}.

Suppose $L$ is an element in $\Omega^1(N)\otimes \Omega_1(N)$ on an $n$-manifold $N$. Then we can
treat $L$ as a pointwise linear map $\Omega^1(N)\to \Omega^1(N)$.  By linear algebra, it follows
that around each $p\in N$ there are coordinates such that at $p$, components $(L^j_i)_{ij}$ is a
matrix in Jordan normal form. Since there are only finitely many ways an $n\times n$ matrix can be
decomposed into Jordan blocks, it follows that there are only a finite number of normal forms for
$L\vert_p$.  It should be emphasised that the structure of the Jordan normal form 
is unstable under perturbations of the matrix. Hence, the normal form is in general only valid at
one point.  The normal form theorem in \cite{Schuller:2010} is essentially an analogous result for
skewon-free elements $\kappa$ in $\medTensor(N)$. The difficulty in proving such a result is easy to
understand.  The matrix that represents $\kappa$ at a point is a $6\times 6$ matrix. By a linear
transformation in $\setR^6$, we can transform this into an Jordan normal form, but such a
transformation, \emph{a priori} has $36$ degrees of freedom. On the other hand, for a coordinate
transformation on $N$, the Jacobian only has $16$ degrees of freedom.  It is therefore not obvious
that coordinate transformations have enough degrees of freedom to transform $\kappa$ into a normal
form. See equation \eqref{eq:kappaTransRuleII}.  For a further discussion, see \cite{Schuller:2010,
  Dahl:2011:Restatement}.

The below theorem summarises the normal form theorem in \cite{Schuller:2010} specialised to the
setting that we need here. Let us make three comments.  First, the below theorem is formulated for
twisted $\kappa\in \medTwTensor(N)$ instead of for \emph{area metrics} in \cite{Schuller:2010}
(which are ordinary tensors) or untwisted $\kappa\in \medTensor(N)$ in \cite{Dahl:2011:Restatement}.
Second, the theorem contains the technical assumption that $\kappa$ is invertible and the Fresnel
surface has no $2$-dimensional subspace. This greatly simplifies the result since it implies that
there are only $7$ possible normal forms and one does not need any conjugations by Hodge operators.
These assumptions will also appear in Theorem \ref{thm:MainDecomp}. For a further discussion of
these assumptions, see end of Section \ref{sec:DoubleCones}. Third, the reason the normal form
theorem is useful can be seen from Proposition \ref{theorem:Decomp}. Namely, in arbitrary
coordinates, a skewon-free $\kappa\in \medTwTensor(N)$ depends on $21$ parameters. However, from
Theorem \ref{thm:classification} we see that each normal form depends only on $2,4$ or $6$
parameters. This reduction of parameters will make the computer algebra feasible in Theorem
\ref{thm:MainDecomp}.

The division into metaclasses in \cite{Schuller:2010} is based on the Jordan block structure of the
matrix representation of $\kappa$ at a point. Since this structure is unstable under perturbations,
it can be difficult to determine the metaclass both in the numerical case and the symbolic case
\cite{LiZhangWang:1997}.

\begin{theorem}
\label{thm:classification} 
Suppose $N$ is a $4$-manifold and  $\kappa\in \medTwTensor(N)$. If $p\in N$ and 
\begin{enumerate}
\item[\emph{(a)}] 
 $\kappa$ has no skewon part at $p$, 
\item[\emph{(b)}]  $\kappa$ is invertible at $p$,
\item[\emph{(c)}]   the Fresnel surface $F_p(\kappa)$ does not contain a two dimensional vector subspace.
\end{enumerate}
Then there exists coordinates $\{x^i\}_{i=0}^3$ around $p$ such that the $6\times 6$ matrix
$(\kappa_I^J)_{IJ}$ that represents $\kappa\vert_p$ in these coordinates is one of the below
matrices:
\begin{itemize}
\item Metaclass I: 
\begin{eqnarray}
\label{eq:MetaClassI}
\begin{pmatrix}
\rr_1 & 0 & 0 & -\pos_1 & 0 &0 \\ 
0 & \rr_2 & 0 & 0  &-\pos_2 & 0 \\
0 & 0 &  \rr_3 &0 & 0 &-\pos_3  \\
\pos_1 & 0 &  0 & \rr_1 & 0 &0  \\
0 & \pos_2 & 0 & 0  &\rr_2 & 0 \\
0 & 0 & \pos_3 & 0 & 0 &\rr_3 
\end{pmatrix}
\end{eqnarray}

\item Metaclass II: 
\begin{eqnarray}
\label{eq:MetaClassII}
\begin{pmatrix}
\rr_1 &       -\pos_1       &    0 & 0 & 0 & 0 \\
\pos_1      &    \rr_1    &    0&  0  & 0 & 0 \\
0              &                0 &  \rr_2 & 0 &0 &-\pos_2 \\
0              &                1 & 0 & \rr_1 & \pos_1 &0  \\
1              &                0 & 0 & -\pos_1  &\rr_1 & 0 \\
0              &                0 & \pos_2 & 0 & 0 &\rr_2 
\end{pmatrix}
\end{eqnarray}

\item Metaclass III: 
\begin{eqnarray}
\label{eq:MetaClassIII}
\begin{pmatrix}
\rr_1         &       -\pos_1       &    0 & 0 & 0 & 0 \\
\pos_1      &    \rr_1           &    0&  0  & 0 & 0 \\
1              &                0    &  \rr_1 & 0 &0 &-\pos_1 \\
0              &                0    & 0 & \rr_1 & \pos_1 &1  \\
0              &                0    & 1 & -\pos_1  &\rr_1 & 0 \\
0              &                1    & \pos_1 & 0 & 0 &\rr_1
\end{pmatrix}
\end{eqnarray}
\item Metaclass IV:  
\begin{eqnarray}
\label{eq:MetaClassIV}
\begin{pmatrix}
\rr_1 & 0 & 0 & -\pos_1 & 0 &0 \\
0 & \rr_2 & 0 & 0  &-\pos_2 & 0 \\
0 & 0 &  \rr_3 & 0 & 0 &\rr_4  \\
\pos_1 & 0 &  0 & \rr_1 & 0 &0  \\
0 & \pos_2 & 0 & 0  &\rr_2 & 0 \\
0 & 0 & \rr_4 & 0 & 0 &\rr_3
\end{pmatrix}
\end{eqnarray}
\item Metaclass V: 
\begin{eqnarray}
\begin{pmatrix}
\label{eq:MetaClassV}
\rr_1 &       -\pos_1           &           0 & 0 & 0 & 0 \\
\pos_1      &    \rr_1           &           0&  0  & 0 & 0 \\
0              &                0     &     \rr_2 & 0 &0 &\alpha_3 \\
0              &                1    &           0 & \rr_1 & \pos_1 &0  \\
1              &                0    &          0 & -\pos_1  &\rr_1 & 0 \\
0              &                0    &  \alpha_3 & 0 & 0 &\rr_2 
\end{pmatrix}
\end{eqnarray}
\item Metaclass VI: 
\begin{eqnarray}
\label{eq:MetaClassVI}
\begin{pmatrix}
\rr_1 & 0 & 0 & -\pos_1 & 0 &0 \\
0 & \rr_2 & 0 & 0  &\rr_4 & 0 \\
0 & 0 &  \rr_3 & 0 & 0 &\rr_5  \\
\pos_1 & 0 &  0 & \rr_1 & 0 &0  \\
0 & \rr_4 & 0 & 0  &\rr_2 & 0 \\
0 & 0 & \rr_5 & 0 & 0 &\rr_3 
\end{pmatrix}
\end{eqnarray}
\item Metaclass VII: 
\begin{eqnarray}
\label{eq:MetaClassVII}
\begin{pmatrix}
\rr_1 & 0 & 0 & \rr_4 & 0 &0 \\
0 & \rr_2 & 0 & 0  &\rr_5 & 0 \\
0 & 0 &  \rr_3 & 0 & 0 &\rr_6  \\
\rr_4 & 0 &  0 & \rr_1 & 0 &0  \\
0 & \rr_5 & 0 & 0  &\rr_2 & 0 \\
0 & 0 & \rr_6 & 0 & 0 &\rr_3 
\end{pmatrix}
\end{eqnarray}
\end{itemize}
In each matrix the 
 parameters satisfy $\alpha_1,\alpha_2, \ldots\in \setR$,
$\beta_1, \beta_2, \ldots \in \setR\slaz$ and $\operatorname{sgn}\beta_1 = \operatorname{sgn}\beta_2 = \cdots$.
\end{theorem}

\begin{proof}
  Let $(U,x^i)$ be coordinates around $p$, and let $\mediumMatrix=(\kappa_{I}^J)_{IJ}$ be the $6\times 6$-matrix
  that represents $\kappa$ at $p$ in these coordinates.
By treating $U$ as a manifold with coordinates $\{x^i\}_{i=0}^3$,
  equation \eqref{eq:kappaLocal} defines a tensor $\kappa \in\medTensor(U)$.
  Since $\kappa$ is invertible at $p$ and $F_p(\kappa)$ has no $2$-dimensional subspace, the Jordan
  normal form of $\mediumMatrix$ can not have a Jordan block of dimension $2,\ldots ,6$ that
  corresponds to a real eigenvalue of $\mediumMatrix$.  For area metrics this is established in
  Lemma 5.1 in \cite{Schuller:2010}. (Or, for a translation
to elements in $\medTensor(U)$, see the proof of Theorem 2.1 in
  \cite{Dahl:2011:DoubleCone}.)  In the terminology of \cite{Schuller:2010} and
  \cite{Dahl:2011:DoubleCone} this implies that $\kappa\vert_p$ is of Metaclasses I, $\ldots$,
  VII. Hence Theorem 3.2 in \cite{Dahl:2011:Restatement} (the restatement of the normal form theorem
  in \cite{Schuller:2010}) 
  implies that around $p$, manifold $U$ has a coordinate chart $(\widetilde U, \widetilde x^i)$
such that at $p$, we have
\begin{eqnarray}
\label{eq:transformationToJ}
T \mediumMatrix T^{-1}
&=&R,  
\end{eqnarray}
%
where $T=(\pd{x^J}{\widetilde x^I})_{IJ}$ 
is as in equation
\eqref{eq:antiSymJacobian} and $R$ is one of the $6\times 6$ matrices in equations
\eqref{eq:MetaClassI}---\eqref{eq:MetaClassVII}
for some parameters $\alpha_1, \alpha_2, \ldots \in \setR$ and $\beta_1, \beta_2, \ldots>0$.
Since $(U, x^i)$ is a chart in $N$ it follows that $(\widetilde U, \widetilde x^i)$ is also a chart
in $N$.  Multiplying equation \eqref{eq:transformationToJ} by $\operatorname{sgn} \det \left(
  \pd{x^i}{\widetilde x^j}\right)$ and comparing with equation \eqref{eq:kappaTransRuleII} shows
that $\operatorname{sgn} \det \left( \pd{x^i}{\widetilde x^j}\right) R$ is the 
matrix
that represents $\kappa\in \medTwTensor(N)$ in coordinates $\{\widetilde x^i\}_{i=0}^3$.
If  $\operatorname{sgn} \det \left( \pd{x^i}{\widetilde x^j}\right) =1$ or if $R$ is in Metaclasses I, IV, VI, VII, the claim follows.   
On the other hand, if $\operatorname{sgn} \det \left( \pd{x^i}{\widetilde x^j}\right) =-1$ and $R$
is in Metaclasses II, III, V, it remains to prove that we can change the signs of the $1$-entries in
the normal forms by an orientation preserving coordinate transformation.  Let $\{\widehat
x^i\}_{i=0}^3$ be coordinates determined by $\widehat x^i = J^i_j \widetilde x^j$ for a suitable
$4\times 4$ matrix $J=(J^i_j)_{ij}$.
For Metaclass III a suitable Jacobian is $(J^i_j)_{ij} = \operatorname{diag}(1,-1,-1,1)$, and
for Metaclass II and V a suitable Jacobian is
\begin{eqnarray*}
J &=& 
\begin{pmatrix}
1 &0 &0 &0\\
0 &0 &-1 &0\\
0 &1 &0& 0\\
0 &0 &0 &1\\
\end{pmatrix}.
\end{eqnarray*}
\qedhere
\end{proof}

\subsection{Non-birefringent medium}
\label{sec:nonbire}
By a \emph{pseudo-Riemann metric} on a manifold $N$ we mean a symmetric $0\choose 2$-tensor $g$ that
is non-degenerate. If $N$ is not connected we also assume that $g$ has constant signature.
By a \emph{Lorentz metric} we mean a pseudo-Riemann metric on a $4$-manifold with signature $(-+++)$
or $(+---)$.  
Let $\sharp$ be the isomorphisms $\sharp\colon T^\ast N\to
TN$, so that if locally $g=g_{ij} dx^i \otimes dx^j$ then $\sharp(\alpha_i dx^i)= \alpha_i g^{ij}\pd{}{x^j}$.
Using the $\sharp$-isomorphism we extend $g$ to covectors by setting
$g(\xi,\eta)=g(\xi^\sharp,\eta^\sharp)$ when $\xi,\eta\in T^\ast_p(N)$.

For a Lorentz metric $g$ the  \emph{light cone} at a point $p\in N$ is defined as
\begin{eqnarray*}
N_p(g) &=& \{\xi \in T_p^\ast(N) : g(\xi,\xi) = 0\}, 
\end{eqnarray*}
and analogously to the Fresnel surface we define $N(g)=\coprod_{p\in N} N_p(g)$.

If $g$ is a pseudo-Riemann metric on a $4$-manifold $N$, then the \emph{Hodge star operator} of $g$
is defined as the $\ast_g \in \medTwTensor(N)$ such that 
if locally
$g=g_{ij}dx^i\otimes dx^j$, and $\ast_g$ is written as in equation \eqref{eq:kappaLocal}, then
\begin{eqnarray}
\label{eq:hodgeKappaLocal}
(\ast_g)^{ij}_{rs} &=& \sqrt{\vert \det g\vert}\, g^{ia}g^{jb} \varepsilon_{abrs},
\end{eqnarray}
where $\det g= \det g_{ij} $ and $g^{ij}$ is the $ij$th entry of $(g_{ij})^{-1}$. Then $\ast_g$ has
only a principal part.  See for example, \cite{Obu:2003, Favaro:2011}.
Moreover, if $g$ is a Lorentz metric and $\kappa=\ast_g$, we have
\begin{eqnarray}
\label{eq:kappaNonBire}
  F(\kappa) &=& N(g).
\end{eqnarray}
Equation \eqref{eq:kappaNonBire}
is the motivation for defining $N(g)$ as a subset of the cotangent bundle.

\begin{definition}
\label{def:nonbirefringent}
Suppose $N$ is a $4$-manifold and $\kappa\in \medTwTensor(N)$. Then $\kappa$ 
is \emph{non-birefringent} if there exists a Lorentz metric $g$ on $N$ such that
equation \eqref{eq:kappaNonBire} holds.
\end{definition}

Thus, in non-birefringent medium, the Fresnel surface $F_p(\kappa)$ has only a single sheet, and
there is only one signal speed in each direction. In non-birefringent medium it follows that 
propagation speed can not depend on polarisation. 
On $N=\setR^4$, a specific example of a non-birefringent medium is 
$\kappa = \sqrt{\frac{\epsilon}{\mu}} \ast_g$,
where $g$ is the Lorentz metric $g= \operatorname{diag}(-\frac{1}{\epsilon \mu},1,1,1)$ on
$\setR^4$. Then constitutive equation \eqref{FGchi} models standard isotropic
medium on $\setR^4$ with permittivity $\epsilon>0$ and $\mu>0$.
The next theorem gives the complete characterisation of all non-birefringent media with only a only
a principal part.

\begin{theorem}
\label{thm:mainResult}
Suppose $N$ is a $4$-manifold.  If $\kappa\in \medTwTensor(N)$ satisfies $\kappaII=0$, then the
following conditions are equivalent:
\begin{enumerate}
\item 
\label{coIII} 
$\kappaIII=0$ and $\kappa$ is non-birefringent.

\item $\kappa^2 = -f \operatorname{Id}$ for some function $f\in C^\infty(N)$ with $f>0$.
\label{coI}
\item 
 \label{coII} 
 there exists a Lorentz metric $g$ and a non-vanishing function $f\in C^\infty(N)$ such that
\begin{eqnarray}
  \kappa &=& f \ast_g.
\end{eqnarray}

\end{enumerate}
\end{theorem}

Implication \ref{coIII} $\Rightarrow$ \ref{coI} was conjectured in 1999 by Y.~Obukhov and F.~Hehl
\cite{ObukhovHehl:1999, ObuFukRub:00}. Under some additional technical assumptions the implication
was already proven in \cite{ObuFukRub:00}.  However, the general case was only established in
\cite{FavaroBergamin:2011} by A.~Favaro and L.~Bergamin by a case by case analysis using the normal
form theorem in \cite{Schuller:2010}.
For an alternative proof using a Gr\"obner basis, see \cite{Dahl:2011:Closure} and 
for similar results, see \cite{
 LamHeh:2004, 
 Itin:2005, 
RSS:2011} 
and Section \ref{sec:mediumWithDoubleLightCone} below. 
Implication \ref{coII} $\Rightarrow$ \ref{coIII} is a direct computation.
In the setting of electromagnetics, implication \ref{coI} $\Rightarrow$ \ref{coII} seems to first to
have been derived by M.~Sch\"onberg \cite{Rubilar2002, Shoenberg:1971}. For further derivations and
discussions, see \cite{Obu:2003, Rubilar2002, ObuFukRub:00, ObukhovHehl:1999, Jadczyk:1979}.
%


When a general $\kappa\in \medTwTensor(N)$ on a $4$-manifold $N$ satisfies $\kappa^2 = -f
\operatorname{Id}$ for a function $f\in C^\infty(N)$ one says that $\kappa$ satisfies the
\emph{closure condition}. For physical motivation, see \cite[Section D.3.1]{Obu:2003}.  For a study
of more general closure relations, and in particular, for an analysis when $\kappa$ might have a
skewon part, see \cite{Favaro:2011, LindellBergaminFavaro:Decomp}, and Section \ref{sec:44Factorisability}
below.

\subsection{Medium with a double light cone}
\label{sec:mediumWithDoubleLightCone}
Since the Fresnel surface is a $4$th order surface, the Fresnel surface can decompose into two
distinct Lorentz null cones. In such medium differently polarised waves can propagate with different
wave speeds.  This is, for example, the case in \emph{uniaxial crystals} like calcite \cite[Section 15.3]{BornWolf:1999}.
This motivates the next definition.

\begin{definition}
\label{def:bire}
Suppose $N$ is a $4$-manifold and  $\kappa\in \medTwTensor(N)$.
If $p\in N$ we say
that the Fresnel surface $F_p(\kappa)$ \emph{decomposes into a double light cone} if
there exists Lorentz metrics $g_+$ and $g_-$ defined in a
neighbourhood of $p$ such that 
\begin{eqnarray}
\label{eq:cFdecomp}
F_p(\kappa)  &=& N_p(g_+)\,\, \cup \,\, N_p(g_-)
\end{eqnarray}
and $N_p(g_+)\neq N_p(g_-)$. 
\end{definition}

If $g,h$ are Lorentz metrics, then $N_p(g)\subset N_p(h)$ implies that at $p$ we have $g=C h$ for
some $C\in \setR\slaz$. See for example \cite{Toupin:1965}.  Thus, if $\kappa$ decomposes into a
double light cone, then $\kappa$ is not non-birefringent.

Under some assumptions, the next theorem 
gives the complete pointwise description of all medium tensors with a double light cone. The theorem
generalises the result in \cite{Dahl:2011:DoubleCone} to twisted tensors.

\begin{theorem}
\label{thm:factorMedium}
Suppose $N$ is a $4$-manifold and  $\kappa\in \medTwTensor(N)$. Furthermore, 
suppose that at some $p\in N$ 
\begin{enumerate}
\item[\emph{(a)}]
 $\kappa$ has no skewon part at $p$, 
\item[\emph{(b)}] $\kappa$ is invertible at $p$, 
\item[\emph{(c)}] the Fresnel surface $F_p(\kappa)$ factorises into a double light cone at $p$.
\end{enumerate}
Then exactly one of the below three possibilities holds:
\begin{enumerate}
\item \textbf{Metaclass I.} 
  There are coordinates $\{x^i\}_{i=0}^3$ around $p$ such that the matrix $(\kappa_I^J)_{IJ}$ that
  represents $\kappa\vert_p$ in these coordinates is given by equation \eqref{eq:MetaClassI} for
  some $\alpha_1, \alpha_2, \alpha_3 \in \setR$ and $\beta_1, \beta_2,\beta_3\in \setR\slaz$ with
\begin{center}
$   \alpha_2 = \alpha_3, \quad
 \beta_2 = \beta_3, \quad
\operatorname{sgn}\beta_1 =\operatorname{sgn}\beta_2 =\operatorname{sgn}\beta_3$
\end{center}
and either $\alpha_1\neq\alpha_2$ or $\beta_1\neq \beta_2$ or both inequalities hold.

\item \textbf{Metaclass II.}
  There are coordinates $\{x^i\}_{i=0}^3$ around $p$ such that the matrix $(\kappa_I^J)_{IJ}$ that
  represents $\kappa\vert_p$ in these coordinates is given by equation \eqref{eq:MetaClassII} for
  some $\rr_1,\rr_2\in \setR$ and $\beta_1, \beta_2\in \setR\slaz$ with 
\begin{eqnarray*}
  \alpha_1=\alpha_2, \quad
  \beta_1=\beta_2.
\end{eqnarray*}
 
\item 
\textbf{Metaclass IV.}
There are coordinates $\{x^i\}_{i=0}^3$ around $p$ such that the matrix $(\kappa_I^J)_{IJ}$ that
represents $\kappa\vert_p$ in these coordinates is given by equation \eqref{eq:MetaClassIV} for some
$\alpha_1, \alpha_2, \alpha_3, \alpha_4 \in \setR$ and $\beta_1, \beta_2\in \setR\slaz$ with
\begin{eqnarray*}
 \alpha_1 = \alpha_2, \quad
 \beta_1 = \beta_2, \quad 
\alpha_4\neq 0,\quad 
\alpha_3^2\neq \alpha_4^2.
\end{eqnarray*}
\end{enumerate} 
Conversely, if $\kappa$ is defined by one of the above three possibilities, then the Fresnel surface
of $\kappa$ decomposes into a double light cone at $p$.
\end{theorem}
 
\begin{proof}
  For $\kappa\in \medTensor(N)$ the result is proven in \cite[Theorem 2.1]{Dahl:2011:DoubleCone} (up
  to a permutation of coordinates in Metaclass I).  The generalisation to $\kappa\in
  \medTwTensor(N)$ follows by the same argument used to prove Theorem \ref{thm:classification}.
  The converse direction can be verified by computer algebra using the explicit Lorentz metrics
  given in \cite{Dahl:2011:DoubleCone}.
\end{proof}
  
In Theorem \ref{thm:factorMedium}, uniaxial medium is given by Metaclass I when
$\alpha_1=\alpha_2=\alpha_3=0$. The main conclusion of the theorem is that there are two (and only
two) additional classes of medium where the Fresnel surface decomposes (Metaclasses II and
IV). In all three classes, there are explicit formulas for the Lorentz metrics that factorise the
Fresnel surface. For a further discussion of these metrics, see \cite{Dahl:2011:DoubleCone}.

In Theorem \ref{thm:MainDecomp} we will show that under suitable assumptions every skewon-free medium
with a double light cone can be written as in equation \eqref{eq:DoubleLightConeExpFirstTake}. 
This medium class is a special class of \emph{generalised  $Q$-medium}  introduced by
I.~Lindell and H.~Wall\'en  in \cite{LindellWallen:2002}.
For further discussions of this medium class, see \cite{LindellWallen:2004,
 Favaro:2011, LindellBergaminFavaro:Decomp}.

\begin{proposition} 
\label{prop:kappaSingleDouble}
Suppose $N$ is a $4$-manifold, $g$ is a Lorentz metric,
$\rho$ is a twisted scalar density of weight $1$, $A\in \Omega_2(N)$ and
$C_1\in \setR\slaz$ and $C_2\in \setR$. Moreover, suppose $\kappa\in \medTwTensor(N)$ is
defined as
\begin{eqnarray}
\label{eq:DoubleLightConeExpFirstTake}
 \kappa &=& C_1 \ast_g + \rho \,\overline{A}\otimes A + C_2 \operatorname{Id}.
\end{eqnarray}
Then $\kappa$ is skewon-free the following claims hold pointwise in $N$:
\begin{enumerate}
\item \label{prop:kappaSingleDouble:I} 
$\kappa$ is non-birefringent if and only if $A=0$ or $\rho =0$.
\item \label{prop:kappaSingleDouble:II}
$\kappa$ has a double light cone if and only if $\rho\neq0$, $A\neq 0$ and 
\begin{eqnarray}
\label{prop:kappaSingleDouble:II:cond}
  \det \kappa &\neq & \left( C_1^2+ C_2^2\right)^2 \,
\left( C_2  + \frac 1 2 \operatorname{trace}(\rho\,\overline A\otimes A)\right)^2.
\end{eqnarray}
\end{enumerate}
\end{proposition}

\begin{proof}
We restrict the analysis to a point $p\in N$, and let $\{x^i\}_{i=0}^3$ be coordinates around 
$p$ such that 
the Lorentz metric has components $g=\pm \operatorname{diag}(-1,1,1,1)$ at $p$.  For claim
\ref{prop:kappaSingleDouble:I}, let us note that the axion component of $\kappa$ does not influence
the Fresnel polynomial. See for example \cite{Obu:2003}. Thus $\kappa$ is non-birefringent when
$A=0$ or $\rho=0$. For the converse direction, suppose $\kappa$ is non-birefringent.  Then Theorem
\ref{thm:mainResult} implies that $(\kappa-\frac 1 6 \operatorname{trace} \kappa\,
\operatorname{Id})^2 = -\lambda \operatorname{Id}$ for some $\lambda >0$. Writing out the last
equation and solving the associated Gr\"obner basis equations (see \cite{CoxLittleOShea:2007,
  Dahl:2011:Closure}) shows that $A=0$ or $\rho = 0$.
For claim \ref{prop:kappaSingleDouble:II}, let us write
$A=\frac 1 2 A^{ij} \pd{}{x^i}\wedge \pd{}{x^j}$. 
Then the Fresnel polynomial at $p$ is given by
\begin{eqnarray}
\label{eq:kappaFresnelDoubleExp}
 \cG^{ijkl}\xi_i \xi_j \xi_k \xi_l &=& - C_1^2 \left(g^{ij} \xi_i \xi_j\right) \, \left(H^{ij} \xi_i \xi_j\right),
\end{eqnarray}
where $g^{ij}=(g^{-1})_{ij}$ and $H^{ij} = C_1 g^{ij} - 2\rho A^{ia} g_{ab} A^{bj}$ (see 
\cite{LindellWallen:2002, LindellBergaminFavaro:Decomp}). Moreover, 
\begin{eqnarray}
\label{eq:kappaDetDoubleCone}
 \det \kappa &=& \left( C_1^2+ C_2^2\right)^2\, 
\left( 
C_1^2+ C_2^2 + E + C_2 \operatorname{trace}(\rho\,\overline A\otimes A)
\right),
\end{eqnarray}
where $E\in \setR$ is an expression that depends on $\rho, C_1$ and $A$.  We will not need the
explicit expression for $E$. However, by computer algebra we see that the same $E$ also appears
in $\det H$ for matrix $H=(H^{ij})_{ij}$. 
Then
equation \eqref{eq:kappaDetDoubleCone} yields
\begin{eqnarray}
 \det H &=& \nonumber - \left( 
C_1^2 + E - \frac 1 4 \left( \operatorname{trace}(\rho\,\overline A\otimes A)\right)^2
\right)^2 \\
&=& \label{eq:lskdjlad09as}
-\left(
\frac{\det \kappa}{ (C_1^2 + C_2^2)^2} - 
\left( C_2 + \frac 1 2\operatorname{trace}(\rho\,\overline A\otimes A)\right)^2
\right)^2.
\end{eqnarray}
If $\kappa$ has a double light cone, claim \ref{prop:kappaSingleDouble:I} implies that $A\neq 0$ and
$\rho \neq 0$. Moreover, by Proposition 1.5 in \cite{Dahl:2011:DoubleCone} and since polynomials
have a unique factorisation into irreducible factors \cite [Theorem 5 in Section
3.5]{CoxLittleOShea:2007}, we have $\det H<0$ and equation \eqref{eq:lskdjlad09as} implies
inequality \eqref{prop:kappaSingleDouble:II:cond} for $\det \kappa$.
Conversely, if the inequalities in claim \ref{prop:kappaSingleDouble:II} are satisfied, then
equation \eqref{eq:lskdjlad09as} shows that $\det H<0$, so $g$ and $H$ both have Lorentz signature
at $p$. To complete the proof we need to show that there is no constant $C\in \setR\slaz$ such that
$g^{ij} = C H^{ij}$. Since $A\neq 0$ and $\rho\neq 0$, this follows by inspecting equations $g^{ii}
= C H^{ii}$ for $i=0,\ldots, 3$.
\end{proof}

\section{Decomposable media}
\label{eq:twoChar}
In this section we first describe the class of decomposable medium introduced in
\cite{LindellBergaminFavaro:Decomp}. In particular, in Theorem \ref{eq:decompRel} we describe the
sufficient conditions derived in \cite{LindellBergaminFavaro:Decomp} that imply that a medium is
decomposable.
In Theorem \ref{thm:MainDecomp} these conditions will play a key role. 
In Section \ref{sec:44Factorisability} we will describe some results that suggest that condition
\ref{eq:decompRel:I} in Theorem \ref{eq:decompRel} is a general factorisability condition
for the Fresnel polynomial.
Following
\cite{LindellBergaminFavaro:Decomp} we restrict the analysis to $\setR^4$ so that we can work with
plane waves.
 
\subsection{Plane waves in $\setR^4$}
\label{sec:PlaneWavesInR4}
We say
that a tensor $T$ on $\setR^4$ is \emph{constant} if there are global coordinates for $\setR^4$
where components for $T$ are constant.  If we assume that many tensors are constant, we assume that
they are constant with respect to the same choice of coordinates.  Below we also use notation
$\Omega^k(N, \setC)$ to denote the space of $k$-forms on a manifold $N$ with possibly complex
coefficients.

Suppose $\kappa\in \medTensor(\setR^4)$ is constant and $F,G\in \Omega^2(\setR^4)$ are defined as
\begin{eqnarray}
\label{eq:FGplanewave}
 F = \operatorname{Re}\{ e^{i\Phi} X\}, \quad
 G = \operatorname{Re}\{ e^{i\Phi} Y\}, 
\end{eqnarray}
where $\Phi$ is a function $\Phi\colon \setR^4\to \setR$ such that $d\Phi$ is constant and non-zero,
$X,Y\in \Omega^2(\setR^4,\setC)$ are constant and not both zero.  If $F$ and $G$ solve the
sourceless Maxwell's equations we say that $F$ and $G$ is a \emph{plane wave}.


\begin{proposition} 
\label{prop:HodgeAxionPW:polarisations}
Suppose $\kappa\in \medTensor(\setR^4)$ is constant and
$\Phi$ is a function
$\Phi\colon \setR^4\to \setR$ such that $d\Phi$ is constant and non-zero.
Moreover, suppose $X,Y$ are constant $2$-forms $X,Y\in \Omega^2(\setR^4,\setC)$.
If $F$ and $G$ are defined by equations \eqref{eq:FGplanewave},
then the following conditions are equivalent:
\begin{enumerate}
\item \label{prop:HodgeAxionPW:polarisations:i}
$F$ and $G$ is a plane wave.
\item \label{prop:HodgeAxionPW:polarisations:ii}
$d\Phi\in F(\kappa)$ and there exists a constant
$\alpha \in \Omega^1(\setR^4,\setC)$ such that $d\Phi\wedge\alpha\neq 0$, $d\Phi\wedge\kappa(d\Phi\wedge\alpha)=0$ and
\begin{eqnarray}
X &=& \label{prop:decompEx:proofX}
d\Phi \wedge \alpha, \\ 
Y &=& \label{prop:decompEx:proofY} 
\kappa(d\Phi\wedge \alpha).
\end{eqnarray} 
\end{enumerate}
\end{proposition}

\begin{proof} 
Let $\xi=d\Phi$.  
If $F$ and $G$ is a plane wave then $\xi\neq 0$ implies that
\begin{eqnarray}
\label{eq:dxiXYkappa}
\xi \wedge X = 0, \quad 
\xi \wedge Y = 0, \quad 
Y=\kappa(X).
\end{eqnarray}
The first equation in equation \eqref{eq:dxiXYkappa} implies that there exists a
constant $1$-form $\alpha\in \Omega^1(\setR^4, \setC)$ such that $X = \xi\wedge \alpha$.  It is
clear that $\alpha$ and $\xi\wedge \alpha$ are both non-zero, since otherwise $X=Y=0$.  Combining
the latter two equations in equation \eqref{eq:dxiXYkappa} 
implies that
\begin{eqnarray}
\label{eq:xiVkappaxiAlpha=0}
\xi\wedge\kappa(\xi\wedge\alpha)&=&0.
\end{eqnarray} 
Since this linear equation for $\alpha$ has a non-zero solution, it follows that 
$\xi\in F(\kappa)$. See for example, 
\cite{ObuFukRub:00, Rubilar2002, Obu:2003,  Dahl:2011:Closure}.
This completes the proof of implication 
\ref{prop:HodgeAxionPW:polarisations:i} $\Rightarrow$
\ref{prop:HodgeAxionPW:polarisations:ii}.
For the converse implication it suffices to verify that equations
\eqref{eq:FGplanewave}--\eqref{prop:decompEx:proofY} define a solution to Maxwell's equations.
\end{proof}

\subsection{Decomposable medium}
\label{sec:decompo}
The next definition and theorem  are from \cite{LindellBergaminFavaro:Decomp}.
It is not known if the converse of  Theorem \ref{eq:decompRel} is also true
\cite{LindellBergaminFavaro:Decomp}.


%

\begin{definition}
\label{def:decomposables}
Suppose $\kappa \in \medTensor(\setR^4)$ is constant. Then we say that $\kappa$ is
\emph{decomposable} if there exist non-zero and constant $A,B\in \Omega_2(\setR^4)$ such that
if $F,G$ is a plane wave solution to Maxwell's equations, then
\begin{eqnarray}
\label{eq:FAFB0}
 F(A) = 0 \quad\mbox{or}\quad    F(B)=0.
\end{eqnarray}
\end{definition}

%

\begin{theorem}
\label{eq:decompRel}
Suppose $\kappa\in \medTensor(\setR^4)$ is constant. Furthermore, suppose
\begin{enumerate}
\item \label{eq:decompRel:I}
there exists constant tensors
$A,B\in \Omega_2(\setR^4)$ and a constant scalar density $\rho$ of weight $1$ such that 
\begin{eqnarray}
\label{eq:eq44}
\alpha \operatorname{Id} + \beta\left( \kappa+ \overline{\kappa} \right)+ \gamma \overline{\kappa}\circ \kappa &=& 
\rho \left(\overline A \otimes B +\overline B\otimes A\right)
\end{eqnarray}
for constants $\alpha, \beta, \gamma\in \setR$ and $\beta,\gamma$ are not both zero.
\item \label{eq:decompRel:II}
the right hand side in equation \eqref{eq:eq44} is non-zero.
\end{enumerate}
Then $\kappa$ is decomposable (and condition \eqref{eq:FAFB0} holds for the same $A$ and $B$
as in condition \eqref{eq:eq44}).
\end{theorem}


Before the proof, let us note that by Lemma \ref{lemma:ABBA=0}, the right hand side in equation
\eqref{eq:eq44} is non-zero if and only if $A,B$ and $\rho$ are all non-zero.

\begin{proof}
(Following \cite{LindellBergaminFavaro:Decomp}.)
Suppose condition \eqref{eq:eq44} holds for some $\alpha, \beta, \gamma,\rho,A,B$.
%
Moreover, suppose $F,G$ is an arbitrary plane wave for $\kappa$ as in equation
\eqref{eq:FGplanewave}.  To prove the claim we need to show that condition \eqref{eq:FAFB0} holds.
Proposition \ref{prop:HodgeAxionPW:polarisations} implies that $Y=\kappa(X)$ and 
\begin{eqnarray*}
 X\wedge X = 0, \quad
 X\wedge Y = 0, \quad
 Y\wedge X = 0, \quad
Y\wedge Y = 0,
\end{eqnarray*}
whence equation \eqref{eq:diamondTranspose} implies that  
\begin{eqnarray}
\label{eq:iou31io}
0&=& X\wedge \left( \alpha \operatorname{Id} + \beta (\kappa + \overline{\kappa})+ \gamma \overline{\kappa} \circ \kappa\right)(X).
\end{eqnarray} 
Let $\{x^i\}_{i=0}^3$ be coordinates for $\setR^4$ where all the aforementioned tensors are constant. Then
\begin{eqnarray*}
0&=& X\wedge 
\rho \left(\overline A \otimes B + \overline B\otimes A \right)
(X) \\
&=&X(A) X(B) \,\, \rho dx^0\wedge dx^1\wedge dx^2\wedge dx^3.
\end{eqnarray*}
Here, the first equality follows by condition \eqref{eq:eq44} and \eqref{eq:iou31io}, and the latter
equality follows by a computation in coordinates.  Since $A$ and $B$ are real, it follows that
$F(A)=0$ or $F(B)=0$.
\end{proof}

In Theorem \ref{thm:MainDecomp} we will see that all the medium tensors in Theorem
\ref{thm:factorMedium} are decomposable. In particular, uniaxial medium is decomposable.
The next proposition shows that isotropic medium determined by
a Hodge star operator is never decomposable.

\begin{proposition} 
\label{prop:decompEx}
Suppose $\kappa\in \medTensor(\setR^4)$ is defined as
\begin{eqnarray*}
\kappa &=& C_1 \ast_g \,\,+\,\, C_2 \operatorname{Id},
\end{eqnarray*}
where $C_1\in \setR\slaz$, $C_2\in \setR$ and $g$ is a constant indefinite pseudo-Riemann metric on $\setR^4$.
Then $\kappa$ is not decomposable. 
\end{proposition}

\begin{proof} 
 Let us first assume that $g$ is a Lorentz metric and let $\{x^i\}_{i=0}^3$ be coordinates such
that $g= k\operatorname{diag}(-1,1,1,1)$ for some $k\in \{-1,1\}$. At $0\in \setR^4$, it follows that
\begin{eqnarray*}
  F_0(\kappa) &=& \{ \xi \in T^\ast_0(\setR^4) : -\xi_0^2 +\xi_1^2 +\xi_2^2 +\xi_3^2 = 0\}.
\end{eqnarray*}
 For a contradiction, suppose $\kappa$ is decomposable. By Proposition
 \ref{prop:HodgeAxionPW:polarisations} there exists a non-zero and constant $A,B\in
\Omega_2(\setR^4)$ such that
\begin{eqnarray} 
\label{eq:AxiValphaYY}
(\xi \wedge \alpha)(A) \,\,   (\xi\wedge\alpha)(B) &=& 0
\end{eqnarray} 
for all $\xi, \alpha \in T^\ast_0(\setR^4)$ that satisfy $\xi\in F_0(\kappa)$ and 
\begin{eqnarray} 
\label{eq:gxixixialphapropTT}
\xi\wedge \alpha \neq 0,\quad \xi\wedge\kappa(\xi\wedge\alpha ) =0.
\end{eqnarray}
%
Let $G$ is the subset $G\subset F_0(\kappa)\slaz$ for which each coordinate belongs to $\{0,1,
\sqrt 2, \sqrt 3\}$. That is, one can think of $G$ as a discretisation of $F_0(\kappa)$ in one
quadrant of $T_0^\ast(\setR^4)$. In total there are $19$ such points, and for each $\xi\in G$, we
can find two linearly independent $\alpha\in T^\ast_0(\setR^4)$ such conditions
\eqref{eq:gxixixialphapropTT} holds, \emph{cf.} \cite{Dahl:2011:Closure}.
Insisting that equation \eqref{eq:AxiValphaYY} holds for all such $\xi$ and $\alpha$ gives $19\times
2=38$ second order polynomial equations for variables in $A$ and $B$. Computing a Gr\"obner basis for these
equations and solving implies that either $A=0$ or
$B=0$. See \cite{CoxLittleOShea:2007}. 
Hence $\kappa$ is not decomposable.
When $g$ has signature $(--++)$ the claim follows by repeating the above argument.
\end{proof}

\subsection{Factorisability of the Fresnel polynomial}
\label{sec:44Factorisability}
In what follows condition \ref {eq:decompRel:I} in Theorem \ref {eq:decompRel} will play a key
role. Let us therefore introduce the following definition.

\begin{definition}
\label{eq:newDecompDef}
If $\kappa\in \medTensor(\setR^4)$ is constant and satisfies condition  \ref{eq:decompRel:I} 
in Theorem \ref{eq:decompRel}, then we say that $\kappa$ is 
\emph{\newdecomp{}}.
\end{definition}

In \cite{LindellBergaminFavaro:Decomp}, I.~Lindell, L.~Bergamin and A.~Favaro showed that if
$\kappa$ is \newdecomp{} (plus some additional assumptions), then the Fresnel polynomial of $\kappa$
always factorises into the product of two quadratic forms.  In this section we summarise this result
in Theorem \ref{thm:LBFFactorTheorem}. Moreover, we will see that 
for  \newdecomp{} medium, the Fresnel polynomial seems to
factorise even when the additional assumptions in Theorem \ref{thm:LBFFactorTheorem} are not
satisfied.  
These results suggest (but do not prove) that the definition of \newdecomp{} medium might be a sufficient
condition for the Fresnel polynomial to factorise. 
%

Let us first note that the class of \newdecomp{} media contains a number medium classes as special
cases.  If $\kappa$ is purely skewon, then $\kappa+\overline{\kappa}=0$ and $\kappa$ is
\newdecomp. Also, if $\kappa$ satisfies the \emph{mixed closure condition} $\overline{\kappa}\circ
\kappa=\lambda \operatorname{Id}$ 
\cite{LindellBergaminFavaro:Decomp, Favaro:2011}, then $\kappa$ is
\newdecomp{}.  
If $\kappa$ has no skewon part, then $\overline{\kappa} = \kappa$ and the definition of \newdecomp{}
medium simplifies. Thus, if $\kappa$ has no skewon part and if $\kappa$ is a \emph{self-dual medium}
(so that $\alpha \operatorname{Id} + \beta \kappa + \gamma \kappa^2=0$) \cite{Lindell:SD:2008}, then
$\kappa$ is  \newdecomp.  In particular, skewon-free medium that satisfies the \emph{closure
 condition} $\kappa^2 = \lambda \operatorname{Id}$ \cite{Obu:2003} is \newdecomp.

 %

Equation \eqref{eq:eq44} that defines \newdecomp{} medium is a nonlinear
equation in $\kappa$. Suppose $\{x^i\}_{i=0}^3$ are coordinates for $\setR^4$,  $\mediumMatrix\in
\setR^{6\times6}$ is the matrix $\mediumMatrix=(\kappa_I^J)_{IJ}$ that represents $\kappa$ and $A,B\in \setR^6$ are the
 column vectors $A=(A^I)_I$ and $B=(B^I)_I$ that represent bivectors $A$ and $B$
with components as in Section \ref{media:decomp}.  Then equation \eqref{eq:eq44} reads
\begin{eqnarray}
\label{eq:eq44matrix}
\alpha E + \beta(\mediumMatrix^t E  + E \mediumMatrix) + \gamma \mediumMatrix^t E \mediumMatrix = 2 \rho (AB^t + BA^t),
\end{eqnarray}
where $A^t$ is the matrix transpose and $E\in \setR^{6\times 6}$ is 
the matrix $E=(\varepsilon^{IJ})_{IJ}$.
Numerically, 
$E= \begin{pmatrix}
0& I\\
I & 0
\end{pmatrix}$,  
where $0$ and $I$ are the zero and identity $3\times 3$ matrices.
When $\gamma\neq 0$, equation \eqref{eq:eq44matrix} is structurally similar to an \emph{algebraic
  Riccati equation} \cite{GohbergLancasterRodman:2005}.

The next theorem summarises the factorisation result from \cite{LindellBergaminFavaro:Decomp}, but
restated in the present setting.

\begin{theorem}
\label{thm:LBFFactorTheorem}
If $\kappa\in \medTensor(\setR^4)$ is \newdecomp{} and 
$\alpha,\beta,\gamma, \rho,A,B$ in equation \eqref{eq:eq44}
satisfy one of the below conditions:
\begin{enumerate}
\item \label{thm:LBFFactorTheorem:I} 
$\gamma = 0$,
\item \label{thm:LBFFactorTheorem:II} $\gamma\neq 0$, $\beta^2 -\alpha\gamma\neq 0$ and 
there exists a
$D\in \Omega_2(\setR^4)$ such that
\begin{eqnarray}
\label{eq:Dcondition}
 D\left( \gamma \kappa + \beta \operatorname{Id}\right) &=& 
\frac 1 2 \operatorname{trace} (\rho\, \overline D\otimes D ) \, A + \gamma B.
\end{eqnarray}
\end{enumerate}
Then the Fresnel polynomial of $\kappa$ factorises into the product of two quadratic forms.
\end{theorem}

Let us note that equation \eqref{eq:Dcondition} is a non-linear equation for $D$.  \emph{A priori},
the equation has real solutions, complex solutions, or no solutions for $D$.  For a discussion of the 
last possibility, see below.
Pointwise $\operatorname{trace} (\rho\, \overline D\otimes D)=0$ holds if and only if $D\wedge D=0$ or
$\rho =0$.

Let us outline the argument in \cite{LindellBergaminFavaro:Decomp} used to prove Theorem
\ref{thm:LBFFactorTheorem}. Suppose $\medTensor(\setR^4)$ is \newdecomp{}.  If assumption
\ref{thm:LBFFactorTheorem:I} holds, then by rescaling we may assume that $\beta=1$. Then, since
$\kappa+\overline \kappa = 2(\kappaI+\kappaIII)$, it follows that 
\begin{eqnarray}
\label{eq:SDMexp}
{\alpha} \operatorname{Id} +  2 (\kappa  - \sigma) &=& 
\rho \left(\overline{A} \otimes B +\overline{B}\otimes A\right) 
\end{eqnarray}
for some $\sigma\in \medTensor(\setR^4)$ with only a skewon part.  This gives an explicit
representation formula for all $\kappa$ that satisfy condition \eqref{eq:eq44} with
$\gamma=0$. Computing the Fresnel polynomial for $\kappa$ shows that it factorises into two
quadratic forms.
On the other hand, when assumption \ref{thm:LBFFactorTheorem:II} holds, then
Theorem \ref{thm:LindellBergaminFavaro} in the below shows that equation \eqref{eq:eq44} transforms into
$\overline\eta\circ \eta =\lambda\operatorname{Id}$ for some $\lambda\neq 0$ by a transformation
similar to completing the square.  Thus, to understand the structure of \newdecomp{} medium
that satisfy assumption \ref{thm:LBFFactorTheorem:II}, we only need to understand the
simpler equation $\overline\eta\circ \eta = \lambda\operatorname{Id}$ with $\lambda\neq 0$.
%
%
In \cite{LindellBergaminFavaro:Decomp} the latter equation is solved (see also \cite{Favaro:2011})
using two explicit representation formulas similar to equation \eqref{eq:SDMexp}.  Using these
representation formulas, the Fresnel polynomial can again be computed, and in both cases
it factorises into a product of quadratic forms.

The next theorem from \cite{LindellBergaminFavaro:Decomp} describes the transformation property of
equation \eqref{eq:eq44} used in the proof of Theorem \ref{thm:LBFFactorTheorem}.  The
proof 
is a direct computation using identities \eqref{eq:barId1}--\eqref{eq:barId4}.
For a general discussion of transformation properties for the matrix algebraic Riccati equation, see
\cite{ChoiParkLee:2010, LererRan:2012}.

\begin{theorem} 
\label{thm:LindellBergaminFavaro}
Suppose $\kappa\in \medTensor(\setR^4)$ is \newdecomp{} such that 
equation \eqref{eq:eq44} holds with $\gamma\neq 0$.
If, moreover, there exists a
$D\in \Omega_2(\setR^4)$ such that equation \eqref{eq:Dcondition} holds,
then $\eta\in \medTensor(\setR^4)$ defined as 
\begin{eqnarray}
\label{eq:RiccatiTransformation}
\eta &=& \gamma \kappa -\rho \overline D\otimes A+ \beta\operatorname{Id}
\end{eqnarray}
satisfies
\begin{eqnarray}
\label{eq:mixedClosurePrime}
\overline{ \eta}\circ \eta &=& (\beta^2-\alpha\gamma) \operatorname{Id}.
\end{eqnarray}
\end{theorem}

Suppose $\kappa$ is \newdecomp{} such that equation \eqref{eq:eq44} holds with $\gamma\neq 0$ and
$\beta^2-\alpha\gamma= 0$.  Now we can not use Theorem \ref{thm:LBFFactorTheorem} do decise whether
the Fresnel polynomial factorises. However, by computer algebra we can find explicit examples of
medium tensors with the above properties.  Preliminary computer algebra experiments using such
expressions suggest that the Fresnel polynomial always seems to factorise when the above assumptions
are met. However, the factorisation seems be qualitatively different. Condition
$\beta^2-\alpha\gamma= 0$ seems to imply a linear factor in the Fresnel polynomial.  For example,
the Fresnel polynomial can factorise into the product of irreducible $1$st and $3$rd order
polynomials.
On the other hand, suppose $\kappa$ is \newdecomp{} such that equation \eqref{eq:eq44} holds with
$\gamma\neq 0$, $\beta^2-\alpha\gamma\neq 0$ and equation \eqref{eq:Dcondition} has no real solution
for $D$. Now we can neither use Theorem \ref{thm:LBFFactorTheorem} do decise whether the Fresnel
polynomial factorises, but we may again construct explicit examples of medium tensors with the above
properties. Using these expressions, preliminary computer algebra experiments suggest that the
Fresnel polynomial also seems to factorise in this case.
In conclusion, these initial observations together with Theorem \ref{thm:LBFFactorTheorem} suggest
that the definition of \newdecomp{} medium could be a sufficient condition for
the Fresnel  polynomial to factorise.

Lastly, let us note that algebraic Riccati equations, and more generally, quadratic matrix
equations, appear in a number of fields.  In view of Theorem \ref{thm:LBFFactorTheorem} and equation
\eqref{eq:eq44matrix}, it is, however, interesting to note that quadratic matrix equations appear in
the study of polynomial factorisation in one variable \cite{BiniGemignani:2005}.  Differential
Riccati equations also appear in the problem of factoring linear partial differential operators of
second and third order \cite{GrigorievSchwarz:2004}.

\section{Characterisation and representation of media\\ with a double light cone}
\label{sec:DoubleCones}

\begin{theorem} 
\label{thm:MainDecomp} 
Suppose $N$ is a $4$-manifold, and $\kappa\in\medTwTensor(N)$ 
is skewon-free and invertible at a point $p\in N$.
Then the following conditions are equivalent:
\begin{enumerate}

\item  
\label{thm:MainDecomp:I}
The Fresnel surface of $\kappa$ decomposes into a double light cone at $p$.

\item  $\kappa$ satisfies conditions:
\label{thm:MainDecomp:II} 
\begin{enumerate}
\item[\emph{(a)}] \label{thm:MainDecomp:A1}
the Fresnel surface $F_p(\kappa)\subset T^\ast_p(N)$ does not contain a two-dimensional
vector subspace.
\item[\emph{(b)}] \label{thm:MainDecomp:A2}
there are $A,B\in \Omega^2(N)$ and a tensor density $\rho$ of weight $1$ such that at $p$ we have
\begin{eqnarray}
\label{eq:mainThmCondition}  
(\kappa + \mu \operatorname{Id})^2 &=& -\lambda \operatorname{Id} + 
\rho \left(\overline{A}\otimes B + \overline{B}\otimes A\right)
\end{eqnarray}    
for some $\mu \in \widetilde C^\infty(N)$ and $\lambda \in C^\infty(N)$. Moreover, 
$A,B,\rho \neq 0$ and $\lambda>0$ at $p$.
\end{enumerate}

\item\label{thm:LBFFactorTheorem:III} 
Around $p$ there is a locally defined Lorentz metric $g$, a
locally defined non-zero twisted scalar density $\rho$ of weight $1$, an $A\in \Omega_2(N)$ that is
non-zero at $p$, and constants $C_1\in \setR\slaz$ and $C_2\in \setR$ such that at $p$,
\begin{eqnarray}
\label{eq:RepresentationOfDoubleCone}
\kappa &=& C_1 \ast_g + \rho\, \overline A\otimes A + C_2 \operatorname{Id},
\end{eqnarray}
and $\kappa$ satisfies inequality \eqref{prop:kappaSingleDouble:II:cond} at $p$.
\end{enumerate}
\end{theorem}
  
As described in the introduction, the above theorem is the main result of this paper. A discussion
of the theorem is postponed to the end of this section.  
 
In the Theorem \ref{thm:MainDecomp} we will use 
the computer algebra
technique of \emph{Gr\"obner bases} \cite{CoxLittleOShea:2007} to eliminate variables
from polynomial equations. This technique was also used in \cite{Dahl:2011:DoubleCone}.
Let $\setC[u_1, \ldots, u_N]$ the ring of complex coefficient polynomials $\setC^N\to \setC$ 
in variables $u_1, \ldots, u_N$.   For polynomials $r_1, \ldots, r_k\in \setC[u_1, \ldots, u_N]$, let 
\begin{eqnarray*}
    \langle r_1, \ldots, r_k \rangle=\{ \sum_{i=1}^k f_i r_i : f_i\in \setC[u_1, \ldots, u_N]\}
\end{eqnarray*} 
be the the \emph{ideal generated by $r_1, \ldots, r_k$}.  Suppose $V\subset \setC^N$ is the solution
set to polynomial equations $p_1=0, \ldots, p_M=0$ where $p_i\in \setC[u_1, \ldots, u_N]$.  If $I$
is the ideal generated by $p_1, \ldots, p_M$, the \emph{elimination ideals} are the ideals defined
as
\begin{eqnarray*}
\label{eq:elimId}
      I_k&=& I\cap \setC[u_{k+1}, \ldots, u_{N}],\quad  k\in \{0,\ldots, N-1\}.
\end{eqnarray*} 
Thus, if $(u_1, \ldots, u_N)\in V$ then by \cite[Proposition 9, Section 2.5]{CoxLittleOShea:2007} it
follows that $p(u_{k+1}, \ldots, u_N)=0$ for any $p\in I_k$, and $I_k$ contain polynomial
consequences of the original equations that only depend on variables $u_{k+1}, \ldots, u_N$.  Using
Gr\"obner basis, one can explicitly compute $I_k$ \cite[Theorem 2 in Section
3.1]{CoxLittleOShea:2007}. In the below proof this has been done with the built-in Mathematica
routine {\small \textsf{'GroebnerBasis'}}.
 The same technique of eliminating variables was also a key part of the
proof of Theorem \ref{thm:factorMedium} in \cite{Dahl:2011:DoubleCone}.

\begin{proof}
Let us first prove implication \ref{thm:MainDecomp:I} $\Rightarrow$ \ref{thm:MainDecomp:II}.  By
\cite[Proposition 1.3]{Dahl:2011:DoubleCone} condition \ref{thm:MainDecomp:I} implies that
$F_p(\kappa)$ has no two dimensional subspace.  By Theorem \ref{thm:factorMedium} we only need to
check three medium classes.

\textbf{Metaclass I.}  If $\kappa\vert_p$ is in Metaclass I, then $\kappa$ can be written as in
equation \eqref{eq:MetaClassI} with conditions on the parameters given by 
Theorem \ref{thm:factorMedium}. Suppose $\alpha_1=\alpha_2$. Then 
Theorem \ref{thm:factorMedium} implies that $\beta_1\neq \beta_2$.
Let
$\rho = \frac 1 2 (\beta_2^2-\beta_1^2)$, 
$\mu= -\alpha_1$,
$\lambda=\beta_2^2$. 
Moreover, let $A$ and $B$ be
bivectors defined as $A=\frac 1 2 A^{ij} \pd{}{x^i}\wedge \pd{}{x^j}$ and similarly for $B$, with
coefficients 
\begin{eqnarray}
\label{eq:ASmetaclassIBivectorA} 
(A^{ij})_{ij} = 
\begin{pmatrix} 
0    &           1   &    0 & 0 \\
    &           0    &    0   &  0 \\
    &                    &  0 & 0\\
    &                    &       & 0
\end{pmatrix}, \quad
(B^{ij})_{ij} = 
\begin{pmatrix} 
0          &    0   &   0 & 0 \\
 &       0    &    0   &  0 \\
      &              &  0 & 1\\
      &                   &      & 0
\end{pmatrix},
\end{eqnarray} 
where subdiagonal terms are determined by antisymmetry. 
For these parameters, computer algebra shows that equation \eqref{eq:mainThmCondition} holds.
On the other hand, if $\alpha_1\neq \alpha_2$, suitable parameters are 
\begin{eqnarray*}
\rho = \frac 1 {8 (\alpha_1-\alpha_2)\beta_1}, \quad
\mu = -\alpha_2, \quad
\lambda=\beta_2^2, \quad
\end{eqnarray*}
and 
\begin{eqnarray*}
(A^{ij})_{ij} &=& 
\begin{pmatrix}
0          &     2(\alpha_1-\alpha_2) \beta_1  &    0 & 0 \\
 &            0    &    0   &  0 \\
      &                    &  0 & (\alpha_1-\alpha_2)^2-\beta_1^2+\beta_2^2 + \sqrt{\sigma}\\
      &                    &       & 0
\end{pmatrix},  
\end{eqnarray*} 
where 
\begin{eqnarray*}
 \sigma = 
\left(  
(\alpha_1-\alpha_2)^2+
(\beta_1-\beta_2)^2
\right)\,
\left(  
(\alpha_1-\alpha_2)^2+
(\beta_1+\beta_2)^2
\right).
%
\end{eqnarray*}
Bivector $B$ is defined by the same formula as for $A$, but by replacing $\sqrt{\sigma}$ 
with $-\sqrt{\sigma}$.

\textbf{Metaclass II.}  
If $\kappa\vert_p$ is in Metaclass II, then $\kappa$ can be written as in
equation \eqref{eq:MetaClassII} with conditions on the parameters given by Theorem
\ref{thm:factorMedium}. 
Suitable parameters are $\rho 
= \beta_1/2$, $\mu = -\alpha_1$, $\lambda = \beta_1^2$ and
\begin{eqnarray}
\label{eq:ASmetaclassIIBivectorA} 
(A^{ij})_{ij} = 
\begin{pmatrix}
0          &        1   &    1 & 0 \\
     &            0    &    0   &  0 \\
          &                    &  0 & 0\\
          &                    &       & 0
\end{pmatrix}, \quad
(B^{ij})_{ij} = 
\begin{pmatrix}
0          &      1   &    -1 & 0 \\
     &       0    &    0   &  0 \\
          &              &  0 & 0\\
          &                   &      & 0
\end{pmatrix}.
\end{eqnarray}

\textbf{Metaclass IV.}  If $\kappa\vert_p$ is of Metaclass IV, then $\kappa$ can be written as in
equation \eqref{eq:MetaClassIV} with conditions on the parameters given by Theorem
\ref{thm:factorMedium}.  If $\alpha_1\neq \alpha_3$, then suitable
parameters are
\begin{eqnarray*}
   \rho = \frac{1}{8 (\alpha_3-\alpha_1)\alpha_4}, \quad
   \mu =  -\alpha_1, \quad
   \lambda =\beta_1^2
\end{eqnarray*}
and
\begin{eqnarray*}
(A^{ij})_{ij} &=& 
\begin{pmatrix}
0          &     0  &    0 & (\alpha_1-\alpha_3)^2+\alpha_4^2+\beta_1^2 + \sqrt{\sigma}\\
  &            0    &    2(\alpha_3-\alpha_1) \alpha_4  &  0 \\
       &                    &  0 & 0 \\
       &                    &       & 0
\end{pmatrix}, 
\end{eqnarray*} 
where 
\begin{eqnarray*}
 \sigma = \left(\alpha_4^2 - (\alpha_3-\alpha_1)^2\right)^2 
+ \beta_1^2 \left(2 \alpha_4^2 + \beta_1^2 +
2(\alpha_1-\alpha_2)^2\right).
\end{eqnarray*}  
and $B$ is defined as in Metaclass I.
%
On the other hand, if $\alpha_1=\alpha_3$, then suitable parameters are 
$\rho = \frac 1 2 (\beta_1^2 + \alpha_4^2)$, $\mu = -\alpha_3$, $\lambda
= \beta_1^2$ 
and
\begin{eqnarray}
\label{eq:metaClassIV-Aeq}
(A^{ij})_{ij} = 
\begin{pmatrix}
0          &        0  &    0 & 0 \\
     &            0    &    1   &  0 \\
          &                    &  0 & 0\\
          &                    &       & 0
\end{pmatrix}, \quad
(B^{ij})_{ij} = 
\begin{pmatrix}
0          &      0   &    0 & 1 \\
     &       0    &    0   &  0 \\
          &              &  0 & 0\\
          &                   &      & 0
\end{pmatrix}. 
\end{eqnarray}
This completes the proof of implication 
\ref{thm:MainDecomp:I} $\Rightarrow$ \ref{thm:MainDecomp:II}.

For the converse implication \ref{thm:MainDecomp:II} $\Rightarrow$ \ref{thm:MainDecomp:I}, suppose that $\kappa$
satisfies the conditions in \ref{thm:MainDecomp:II}.  By Theorem \ref{thm:classification} we may
assume that there are coordinates $\{x^i\}_{i=0}^3$ around $p$ such that at $p$, tensor $\kappa$ is
given by one of the matrices in equations \eqref{eq:MetaClassI}--\eqref{eq:MetaClassVII} for some
parameters as in Theorem \ref{thm:classification}.  Let us consider each of the seven cases
separately.

\textbf{Metaclass I.}  If $\kappa\vert_p$ is in Metaclass I, then 
there are coordinates $\{x^i\}_{i=0}^3$ around $p$ such that $\kappa$ is given 
by equation \eqref{eq:MetaClassI}.
By scaling $A$  and
$B$ we may assume that $\rho\vert_p=1$.  Moreover, writing out 
equation \eqref{eq:mainThmCondition}
and eliminating variables in $A$ and $B$ using a Gr\"obner basis (see above) yields equations that
only involve $\lambda, \mu$ and the parameters in $\kappa$.  The rest of the argument is divided
into three subcases:

\textbf{Case 1.} If $\beta_1=\beta_2=\beta_3$ the Gr\"obner basis equations imply that
$\lambda = \beta_1^2$ and
\begin{eqnarray}
(\alpha_2 + \mu)   (\alpha_3 + \mu) &=&  \label{eq:kja111A} 
0,\\
(\alpha_1 + \mu)   (\alpha_3 + \mu) &=&  0,\\
(\alpha_1 + \mu)   (\alpha_2 + \mu) &=&  \label{eq:kja111C}  0.
\end{eqnarray}
It follows that $\alpha_1, \alpha_2, \alpha_3$ can not be all distinct, and by a coordinate change,
we may assume that $\alpha_2=\alpha_3$.  If $\alpha_1=\alpha_2=\alpha_3$, equation
\eqref{eq:kja111A} implies that $\mu=-\alpha_1$. Then equation \eqref{eq:MetaClassI} implies that
$\kappa = -\beta_1 \ast_g + \alpha_1 \operatorname{Id}$ at $p$, where $g$ is the Hodge star operator
for the locally defined Lorentz metric $g=\operatorname{diag}(-1,1,1,1)$. Then equation
\eqref{eq:mainThmCondition} implies that $\rho\,(\overline A\otimes B + \overline B \otimes A)=0$.
Since this contradicts Lemma \ref{lemma:ABBA=0}, we have $\alpha_1\neq \alpha_2$ and $\kappa$ has a
double light cone at $p$ by Theorem \ref{thm:factorMedium}.

\textbf{Case 2.} If exactly two of $\beta_1, \beta_2, \beta_3$ coincide, then after a coordinate
change we may assume that $\beta_1 \neq \beta_2 = \beta_3$. Then the Gr\"obner basis equations
imply that either $\lambda = \beta_1^2$ or $\lambda = \beta_2^2$.
If $\lambda = \beta_1^2$, the Gr\"obner basis equations imply 
that $\alpha_1=\alpha_2=\alpha_3$ and $\beta_1=\beta_2=\beta_3$.
We may therefore assume that $\lambda = \beta_2^2$.  Then the Gr\"obner basis equations imply that
$\mu=-\alpha_2=-\alpha_3$, and $\kappa$ has a double light cone at $p$ by Theorem
\ref{thm:factorMedium}.

\textbf{Case 3.} If all $\beta_1,\beta_2,\beta_3$ are all distinct, then the Gr\"obner basis
equations imply that
\begin{eqnarray*}
(\beta_2^2 - \lambda) (\beta_3^2 - \lambda) (\alpha_1 + \mu)  &=& 0, \\
(\beta_1^2 - \lambda) (\beta_3^2 - \lambda) (\alpha_2 + \mu)  &=& 0, \\
(\beta_1^2 - \lambda) (\beta_2^2 - \lambda) (\alpha_3 + \mu)  &=& 0, \\
(\beta_1^2 - \lambda) (\beta_2^2 - \lambda) (\beta_3^2 - \lambda) &=&0.
\end{eqnarray*}  
These equations imply that we must have $\lambda = \beta_i^2$ 
and $\mu = -\alpha_i$ for some $i\in
\{1,2,3\}$. If $i=1$ the Gr\"obner basis equations imply that $\alpha_1=\alpha_2=\alpha_3$ and
$\beta_1=\beta_2$.  This contradicts the assumption that all $\beta_i$ are distinct. Similarly,
$i=2$ and $i=3$ lead to contradictions, and Case 3 is not possible.

\textbf{Metaclass II.} If $\kappa\vert_p$ is in Metaclass II, there are coordinates
$\{x^i\}_{i=0}^3$ around $p$ such that $\kappa$ is given by equation \eqref{eq:MetaClassII}.
Writing out equation \eqref{eq:mainThmCondition}
and eliminating variables as in Metaclass I gives equations that only involve variables $\lambda,
\mu$ and the variables in $\kappa$.  Solving these equations give
\begin{eqnarray*}
\mu = -\alpha_2, \quad
\lambda = \beta_2^2, \quad
\beta_1 =\beta_2,\quad
\alpha_1 = \alpha_2,
\end{eqnarray*}
and $\kappa$ has a double light cone at $p$ by Theorem \ref{thm:factorMedium}. 

\textbf{Metaclass III.}  If $\kappa\vert_p$ is in Metaclass III, there are coordinates
$\{x^i\}_{i=0}^3$ around $p$ such that $\kappa$ is given by equation \eqref{eq:MetaClassIII}.
Eliminating variables as in Metaclass I implies
that $\beta_1=0$.  Thus $\kappa\vert_p$ can not be in Metaclass III.

\textbf{Metaclass IV.}  If $\kappa\vert_p$ is in Metaclass IV, there are coordinates
$\{x^i\}_{i=0}^3$ around $p$ such that $\kappa$ 
is given by equation \eqref{eq:MetaClassIV}.
%
We have $\alpha_4 \neq 0$ since otherwise
$\operatorname{span} \{ dx^1\vert_p, dx^2\vert_p\}\subset F_p(\kappa)$. Moreover, since $\kappa$ is
invertible at $p$ it follows that $\alpha_3^2\neq \alpha_4^2$.  Writing out equation
\eqref{eq:mainThmCondition}, eliminating variables as in Metaclass I, and solving implies that
\begin{eqnarray*}
\lambda = \beta_1^2, \quad
\beta_1 =\beta_2,\quad
\mu = -\alpha_1, \quad
\alpha_1 = \alpha_2,
\end{eqnarray*}
and $\kappa$ has a double light cone at $p$ by Theorem \ref{thm:factorMedium}.

\textbf{Metaclass V.}  
If $\kappa\vert_p$ is in Metaclass V, there are coordinates $\{x^i\}_{i=0}^3$
around $p$ such that $\kappa$ is given by equation \eqref{eq:MetaClassV}.
We may assume that $\alpha_3\not= 0$, since otherwise
$\operatorname{span} \{dx^i\vert_p\}_{i=1}^3\subset F_p(\kappa)$.  Eliminating variables as in
Metaclass I, and solving implies the contradiction $\lambda +\alpha_3^2=0$.  
Since $\lambda>0$ it follows that  $\kappa\vert_p$ can not be in Metaclass V.

\textbf{Metaclass VI.}  If $\kappa\vert_p$ is in Metaclass VI, there are coordinates
$\{x^i\}_{i=0}^3$ around $p$ such that $\kappa$ is given by equation
\eqref{eq:MetaClassVI}.
Eliminating variables as in Metaclass I implies that 
\begin{eqnarray*}
\left(\lambda + \alpha_5^2 +(\alpha_3 + \mu)^2\right) \left(\lambda+(\alpha_2-\alpha_4+\mu)^2\right) \left(\lambda+(\alpha_2+\alpha_4+\mu)^2\right) &=& 0.
\end{eqnarray*}
Since $\lambda>0$, it follows that $\kappa\vert_p$ can not be in  Metaclass VI.

\textbf{Metaclass VII.}  If $\kappa\vert_p$ is in Metaclass VII, there are coordinates
$\{x^i\}_{i=0}^3$ around $p$ such that $\kappa$ is given by equation
\eqref{eq:MetaClassVII}.
Eliminating variables as in Metaclass I and solving implies that
\begin{eqnarray*}
\prod_{k=1}^3 \left( \lambda + \alpha_{k+3}^2 + (\alpha_k + \mu)^2\right)\,
&=& 0.
\end{eqnarray*}
Since $\lambda>0$, it follows that $\kappa\vert_p$ can not be in Metaclass VII.  This completes the
proof of implication \ref{thm:MainDecomp:II} $\Rightarrow$ \ref{thm:MainDecomp:I}.

Implication \ref{thm:LBFFactorTheorem:III} $\Rightarrow$ \ref{thm:LBFFactorTheorem:I} is a
restatement of Proposition \ref{prop:kappaSingleDouble}.  To prove implication
\ref{thm:LBFFactorTheorem:I} $\Rightarrow$ \ref{thm:LBFFactorTheorem:III} we proceed as in
implication \ref{thm:LBFFactorTheorem:I} $\Rightarrow$ \ref{thm:LBFFactorTheorem:II} and by Theorem
\ref{thm:factorMedium} we only need to check three medium classes. Also, by Proposition
\ref{prop:kappaSingleDouble} we do not need to prove inequality
\eqref{prop:kappaSingleDouble:II:cond} since it follows form the other conditions in
\ref{thm:LBFFactorTheorem:III} when \ref{thm:LBFFactorTheorem:I} holds.

\textbf{Metaclass I.}  If $\kappa\vert_p$ is in Metaclass I, there are coordinates $\{x^i\}_{i=0}^3$
around $p$ such that $\kappa$ is given by equation \eqref{eq:MetaClassI} with conditions 
on the parameters given by  Theorem \ref{thm:factorMedium}. 
Suppose $\alpha_1 = \alpha_2$. 
Let
$C_1 = -\frac{\beta_2^2}{\Psi \sqrt{\vert \det g\vert}}$, $C_2 = \alpha_2$, $\Psi =
\frac{\beta_2^2}{\beta_1}$ and in coordinates $\{x^i\}$, let $\rho$ be defined by $\rho =
(\beta_2^2-\beta^1_2)/(2\beta_1)$.  Then equation \eqref{eq:RepresentationOfDoubleCone} holds when
$A=\frac 1 2 B^{ij} \pd{}{x^i}\wedge \pd{}{x^j}$ when coefficients $B^{ij}$ are as in equation
\eqref{eq:ASmetaclassIBivectorA} and $g$ is the Lorentz metric $g=g_{ij} dx^i\otimes dx^j$ with
coefficients
\begin{eqnarray}
\label{eq:MetaclassIDefineMetric}
(g_{ij})_{ij} &=& \left(\operatorname{diag}\left(
1,-1,-\frac{\Psi}{\beta_2},-\frac{\Psi}{\beta_2}
\right)\right)^{-1}.
\end{eqnarray}
On the other hand, suppose $\alpha_1\neq \alpha_2$.
Let $\Psi$ be one of the two roots to the quadratic equation
\begin{eqnarray}
\label{eq:PsiQuadraticEquation}
\frac1 {\beta_2}\Psi^2 - D_3 \Psi + \beta_2 
&=& 0,
\end{eqnarray}
where $D_3$ is defined as in \cite[Theorem 2.1 \emph{(i)}]{Dahl:2011:DoubleCone} 
\begin{eqnarray*}
 D_3 &=& \frac{(\alpha_1-\alpha_2)^2 + \beta_1^2+ \beta_2^2}
{\beta_1\beta_2}. 
\end{eqnarray*}
Since $\operatorname{sgn}\beta_1=\operatorname{sgn}\beta_2$, the discriminant of equation
\eqref{eq:PsiQuadraticEquation} is strictly positive. Thus $\Psi\in \setR\slaz$ and
$\operatorname{sgn}\Psi=\operatorname{sgn}\beta_1$.  Let $\Xi\in \setR$ be defined as
\begin{eqnarray*}
\Xi &=& \frac 1 2\left(\beta_1 - \beta_2^2 \frac 1 \Psi\right).
\end{eqnarray*}
Since $\alpha_1\neq \alpha_2$ we see that $\Psi=\frac{\beta_2^2}{\beta_1}$ is not a solution to
equation \eqref{eq:PsiQuadraticEquation} whence $\Xi\neq 0$.
Let $C_1,C_2$ be as in the $\alpha_1=\alpha_2$ case and let
$\rho = \operatorname{sgn}\Xi$.
Then equation \eqref{eq:RepresentationOfDoubleCone} holds when 
$g$ is the Lorentz metric 
given by equation \eqref{eq:MetaclassIDefineMetric}
and 
$A=\frac 1 2 A^{ij} \pd{}{x^i}\wedge \pd{}{x^j}$ is given by 
\begin{eqnarray*}
(A^{ij})_{ij} &=& 
\begin{pmatrix} 
0    &           \sqrt{\vert \Xi\vert}  &    0 & 0 \\
&           0    &    0   &  0 \\
&                    &  0 & \frac{ \alpha_1-\alpha_2}{2\rho\sqrt{\vert\Xi\vert}}\\
&                    &       & 0
\end{pmatrix}.
\end{eqnarray*}  

\textbf{Metaclass II.} 
If $\kappa\vert_p$ is in Metaclass II, there are coordinates $\{x^i\}_{i=0}^3$
around $p$ such that $\kappa$ is given by equation \eqref{eq:MetaClassII} with conditions 
on the parameters given by  Theorem \ref{thm:factorMedium}. 
Let
$C_1 = -\frac{1}{\beta_1\sqrt{\vert \det g\vert}}$,
$C_2 = \alpha_1$ 
and 
$\rho = 1/2$. 
Then equation \eqref{eq:RepresentationOfDoubleCone} holds when 
$A=\frac 1 2 A^{ij} \pd{}{x^i}\wedge \pd{}{x^j}$ is 
as in equation \eqref{eq:ASmetaclassIIBivectorA} 
and $g$ is the Lorentz metric $g=g_{ij} dx^i\otimes dx^j$  with coefficients
\begin{eqnarray}
\label{eq:MetaclassIIDefineMetric}
(g_{ij})_{ij} &=& 
\begin{pmatrix}
-1 & 0 & 0  & \beta_1 \\
0 & -\beta_1 & 0  & 0\\
0 & 0 & -\beta_1  & 0 \\
\beta_1 & 0 & 0  & 0
\end{pmatrix}^{-1}.
\end{eqnarray}

\textbf{Metaclass IV.}  If $\kappa\vert_p$ is in Metaclass IV, there are coordinates
$\{x^i\}_{i=0}^3$ around $p$ such that $\kappa$ is given by equation \eqref{eq:MetaClassIV} with
conditions on the parameters given by Theorem \ref{thm:factorMedium}.
Suppose $\alpha_1 = \alpha_3$. Let $C_1 = \frac{\beta_1}{\Psi\sqrt{\vert \det g\vert}}$, $C_2 =
\alpha_1$, $\Psi=\alpha_4/\beta_1$ and $\rho=(\alpha_4^2+\beta_1^2)/(2\alpha_4)$.  Then equation
\eqref{eq:RepresentationOfDoubleCone} holds when $A=\frac 1 2 B^{ij} \pd{}{x^i}\wedge \pd{}{x^j}$
when $B^{ij}$ are as in equation \eqref{eq:metaClassIV-Aeq} and $g$ is the Lorentz metric $g=g_{ij}
dx^i\otimes dx^j$ with coefficients
\begin{eqnarray}
\label{eq:MetaclassIVDefineMetric}
(g_{ij})_{ij} &=& \left(\operatorname{diag}(1,\Psi,\Psi,-1)\right)^{-1}.
\end{eqnarray}
On the other hand, suppose $\alpha_1\neq \alpha_3$.
Let $\Psi$ be one of the two roots to the quadratic equation
\begin{eqnarray}
\label{eq:PsiQuadraticEquationIV}
\Psi^2 +  D_1 \Psi -1 &=& 0,
\end{eqnarray}
where (see \cite[Theorem 2.1 \emph{(iii)}]{Dahl:2011:DoubleCone}),
\begin{eqnarray*}
  D_1 &=& \frac{(\alpha_2-\alpha_3)^2 + \beta_2^2-\alpha_4^2}{\beta_2\alpha_4}.
\end{eqnarray*}
Then $\Psi\in \setR\slaz$ and since $\alpha_1\neq \alpha_3$ 
equation \eqref{eq:PsiQuadraticEquationIV} implies that
$\Psi\neq \frac{\alpha_4}{\beta_1}$. Thus $\Xi\in \setR\slaz$ when
\begin{eqnarray*}
\Xi &=& \frac 1 2\left(\alpha_4 - \beta_1 \Psi\right).
\end{eqnarray*}
Let $C_1,C_2$ be as in  the $\alpha_1=\alpha_3$ case and let
%
$\rho = \operatorname{sgn}\Xi$.
Then equation \eqref{eq:RepresentationOfDoubleCone} holds when 
$g$ is the Lorentz metric in equation \eqref{eq:MetaclassIVDefineMetric}
and $A$ is the bivector 
$A=\frac 1 2 A^{ij} \pd{}{x^i}\wedge \pd{}{x^j}$ with
coefficients 
\begin{eqnarray*}
(A^{ij})_{ij} &=& 
\begin{pmatrix} 
0    & 0 &    0 & \frac{ \alpha_3-\alpha_1}{2\rho\sqrt{\vert\Xi\vert}}\\
&           0    &               \sqrt{\vert \Xi\vert}    &  0 \\
&                    &  0 & 0 \\
&                    &       & 0
\end{pmatrix}.
\end{eqnarray*}  
This completes the proof of implication
\ref{thm:LBFFactorTheorem:I} $\Rightarrow$ \ref{thm:LBFFactorTheorem:III}.
\end{proof}


Let us first emphasise that the conditions in Theorem \ref{thm:MainDecomp} are written analogously
to the conditions in Theorem \ref{thm:mainResult}.  In each theorem, condition
\ref{thm:LBFFactorTheorem:I} is the dynamical description of the medium, condition
\ref{thm:LBFFactorTheorem:II} is a characterisation of the medium and condition
\ref{thm:LBFFactorTheorem:III} is a general representation formula. Let us also emphasise that in
suitable limits,  condition \eqref{eq:mainThmCondition} in Theorem \ref{thm:MainDecomp} reduces
to the closure condition $\kappa^2 = -\lambda \operatorname{Id}$ in Theorem \ref{thm:mainResult},
and representation formula \eqref {eq:RepresentationOfDoubleCone} in Theorem \ref{thm:MainDecomp}
reduces to $\kappa = f\ast_g$ in Theorem \ref{thm:mainResult}. Let us also emphasise that in both 
theorems, all conditions are tensorial, and do not depend on coordinate expressions.
A difference between the theorems is that 
Theorem \ref{thm:mainResult} is a global result, while Theorem \ref{thm:MainDecomp} is
a pointwise result.

All the mediums in Theorem \ref{thm:MainDecomp} satisfy the technical assumptions in Theorem
\ref{thm:LBFFactorTheorem} with either $D=A$ or $D=B$ when $A$ and $B$ are as in equation
\eqref{eq:mainThmCondition}. 

As described in the introduction, condition \ref{thm:LBFFactorTheorem:II} in Theorem
\ref{thm:MainDecomp} is a slight strengthening of the conditions derived in
\cite{LindellBergaminFavaro:Decomp} (see Theorem \ref{eq:decompRel} in the above).  
Representation formula \eqref{eq:RepresentationOfDoubleCone} in Theorem \ref{thm:MainDecomp} is also
adapted from \cite{LindellBergaminFavaro:Decomp}.  For constant medium tensors on $\setR^4$, Theorem
\ref{thm:MainDecomp} implies that if $\kappa$ is invertible, skewon-free and has a double light
cone, then $\kappa$ is algebraically decomposable, and hence decomposable by
\cite{LindellBergaminFavaro:Decomp} (see Theorem \ref{eq:decompRel}).  In this setting, Theorem
\ref{thm:MainDecomp} explicitly shows that the behaviour of signal-speed imposes a constraint on the
behaviour of polarisation.  This can be seen as somewhat unexpected.  However, the explanation is
that polarisation and signal speeds are not independent for a propagating wave, but constrained
by equation \eqref{eq:xiVkappaxiAlpha=0}.  For a further discussion, see \cite{Dahl:2011:Closure}.
It is also instructive to note that condition \eqref{eq:mainThmCondition} is a second order polynomial
constraint on the coefficients in $\kappa$, but the definition of a double light cone involves the
Fresnel surface, which is a constraint involving third order polynomials of the coefficients in
$\kappa$.  The same phenomenon appears in equivalence \ref{coIII} $\Leftrightarrow$ \ref{coI} in
Theorem \ref{thm:mainResult}.

Part of condition \ref{thm:LBFFactorTheorem:II} is condition \emph{(a)}, that states that
the Fresnel surface of $\kappa$ contains no two dimensional subspace. Let us describe five results
where this condition also appears. First, if the Fresnel surface of a $\kappa\in \medTwTensor(N)$
can be written as $F_p(\kappa)=\{\xi \in T^\ast_p(N) : (g(\xi,\xi))^2=0\}$ for a pseudo-Riemann
metric $g$, then condition \emph{(a)} is satisfied if and only if $g$ has signature $(--++)$. This
follows by a result of J.~Montaldi \cite{Montaldi:2006}.  For example, if
$g=\operatorname{diag}(-1,-1,1,1)$, then $F_p(\kappa)$ contains the $2$-dimensional subspace
$\operatorname{span}\{\pd{}{x^0}+ \pd{}{x^3}, \pd{}{x^1}+ \pd{}{x^2}\}$.
Second, one can prove that condition \emph{(a)} is always satisfied if $\kappa$
decomposes into a double light cone (Proposition 1.3 in \cite{Dahl:2011:DoubleCone}).  
Third, in matter dynamics systems, condition \emph{(a)} can be motivated by the behaviour of energy
\cite{RSS:2011}.  In the terminology of \cite{RSS:2011}, condition \emph{(a)} can be replaced by the
stronger condition that $\kappa$ is \emph{bihyperbolic}.
Fourth, condition \emph{(a)} also appears in the study of the well posedness of Maxwell's equations as an
initial value problem \cite{Schuller:2010}.  
Lastly, in the normal form representation of skewon-free medium tensors in \cite{Schuller:2010},
condition \emph{(a)} simplifies the representation since the condition excludes all but the first 7
coordinate representations. 
See \cite{Schuller:2010} and Section \ref{sec:normalSchuller} in the above.

%
%

When equivalence holds in Theorem \ref{thm:MainDecomp}, there does not seem to be a simple relation
between parameters $C_1, C_2, \rho, A, g$ in equation \eqref {eq:RepresentationOfDoubleCone} and
parameters $\mu, \lambda, \rho,A,B$ in equation \eqref{eq:mainThmCondition}. However, if equation
\eqref {eq:RepresentationOfDoubleCone} holds for an $A$ such that $A\wedge A =0$
(that is, $A$ is \emph{decomposable} or \emph{simple} \cite[p.~185]{Cohn:2005}), 
then 
equation \eqref{eq:mainThmCondition} holds for parameters
\begin{eqnarray*}
   \mu = -C_2, \quad \lambda = -C_1^2, \quad B = A(\ast_g).
\end{eqnarray*}
Using a Gr\"obner basis argument one can show that the tensor $\kappa$ defined by equation
\eqref{eq:MetaClassI} when $\beta_1=\beta_2=\beta_3=1$, $\alpha_1 = 1$ and $\alpha_2=\alpha_3=2$ is
invertible and has a double light cone. However, it can not be written as in equation \eqref
{eq:RepresentationOfDoubleCone} for an $A$ such that $A\wedge A = 0$. 

\subsection*{Acknowledgements}
This work has been supported by the Academy of Finland (project 13132527) and by the Institute of
Mathematics at Aalto University. I would like to thank Luzi Bergamin, Alberto Favaro and Ismo
Lindell for useful discussions on this topic.

\providecommand{\bysame}{\leavevmode\hbox to3em{\hrulefill}\thinspace}
\providecommand{\MR}{\relax\ifhmode\unskip\space\fi MR }
\providecommand{\MRhref}[2]{%
  \href{http://www.ams.org/mathscinet-getitem?mr=#1}{#2}
}
\providecommand{\href}[2]{#2}

\end{document}